\documentclass{LMCS}

\def\dOi{9(4:26)2013}
\lmcsheading%
{\dOi}
{1--23}
{}
{}
{Jan.~18, 2013}
{Dec.~27, 2013}
{}

\ACMCCS{[{\bf Theory of computation}]: Models of
  computation---Computability; Computational complexity and
  cryptography---Complexity classes; [{\bf Mathematics of
      computing}]: Continuous mathematics---Topology---Geometric
  topology}


\usepackage[english]{babel}
\usepackage{amsfonts}
\usepackage{amssymb}
\usepackage{amsmath}
\usepackage{amsthm}
\usepackage{verbatim}
\usepackage{algorithmic}
\usepackage{graphicx}
\usepackage{color}
\usepackage{wrapfig}
\usepackage{hyperref}

\theoremstyle{plain}

\newcommand{\rp}{(\mathbb{R}^n,\beta)}
\newcommand{\rpm}{\mathbb{R}^n,\beta}
\newcommand{\rsq}{\mathbb{R}^n}

\newcommand{\mf}{\mathfrak}
\newcommand{\mi}{\mathit}

\begin{document}

\title[Undecidable First-Order Theories of Affine Geometries]{Undecidable First-Order Theories of Affine Geometries\rsuper*}

\author[A.~Kuusisto]{Antti Kuusisto\rsuper a}	%required
\address{{\lsuper a}University of Wroc\l aw}	%required
\email{antti.j.kuusisto@gmail.com}  %optional
%\thanks{thanks 1, optional.}	%optional

\author[J.~Meyers]{Jeremy Meyers\rsuper b}	%optional
\address{{\lsuper b}Stanford University}	%optional
\email{jjmeyers@stanford.edu}  %optional
\titlecomment{{\lsuper*}This work was partially supported by grant 129761 of the Academy of Finland.}	%optional

\author[J.~Virtema]{Jonni Virtema\rsuper c}	%optional
\address{{\lsuper c}University of Tampere}	%optional
\email{jonni.virtema@uta.fi}  %optional
%\thanks{thanks 3, optional.}	%optional

\begin{abstract}
Tarski initiated a logic-based approach to formal geometry that studies
first-order structures with a ternary \emph{betweenness} relation ($\beta$) and a quaternary
\emph{equidistance} relation $(\equiv)$. Tarski established, inter alia, that
the first-order ($\mathrm{FO}$) theory of $(\mathbb{R}^2,\beta,\equiv)$ is decidable.
Aiello and van Benthem (2002) conjectured that the $\mathrm{FO}$-theory of expansions of $(\mathbb{R}^2,\beta)$ with unary predicates
is decidable. We refute this conjecture by showing that for all $n\geq 2$, the $\mathrm{FO}$-theory of the class of expansions of $\rp$ with just one unary predicate is already $\Pi_1^1$-hard
and therefore not even arithmetical.
We also define a natural and comprehensive class $\mathcal{C}$ of geometric structures $(T,\beta)$, where $T\subseteq \rsq$, and show that
for each structure $(T,\beta)\in\mathcal{C}$, the $\mathrm{FO}$-theory of the class of expansions of $(T,\beta)$ with
a single unary predicate is undecidable.
We then consider classes of expansions of structures $(T,\beta)$ with a restricted unary 
predicate, for example a finite predicate, 
and establish a variety of related undecidability results.
In addition to decidability questions, we
briefly study the expressivities of universal $\mathrm{MSO}$ and weak universal $\mathrm{MSO}$ over
expansions of $(\mathbb{R}^n,\beta)$.
While the logics are incomparable in general, over expansions of $(\mathbb{R}^n,\beta)$, formulae of weak universal $\mathrm{MSO}$
translate into equivalent formulae of universal $\mathrm{MSO}$.
\end{abstract}

\maketitle

\section{Introduction}
Decidability of theories of  structures and 
classes of structures is a
central topic in various different fields of computer science and
mathematics, with different motivations and objectives
depending on the field in question.
In this article we investigate formal theories of \emph{geometry} in the
framework introduced by Tarski \cite{Tarski:1948, TarskiGivant:1999}.
The logic-based framework was originally presented in a series of lectures 
given in Warsaw in the 1920's.
The system is based on first-order structures with two
predicates: a ternary \emph{betweenness} relation $\beta$ and a quaternary
\emph{equidistance} relation $\equiv$.
Within this system, $\beta(u,v,w)$ is interpreted to mean that the point
$v$ is between the points $u$ and $w$,
while $x y\equiv u v$ means that the
distance from $x$ to $y$ is equal to the distance from $u$ to $v$.
The betweenness relation $\beta$
can be considered to simulate the action of a ruler,
while the equidistance relation $\equiv$ simulates the action of a compass.
See \cite{schw} and \cite{TarskiGivant:1999} for further information about the
history and development of Tarski's geometry.

Tarski established in \cite{Tarski:1948}
that the first-order theory of
$(\mathbb{R}^2,\beta,\equiv)$ is decidable.
In \cite{vanBenthemAiello:2002}, Aiello and van Benthem
pose the question: ``\emph{What is the complete monadic $\Pi^1_1$ theory of the affine real plane}?''
By \emph{affine real plane}, the authors refer to the structure $(\mathbb{R}^2,\beta)$.
The monadic $\Pi_1^1$-theory of $(\mathbb{R}^2,\beta)$ is of course
essentially the same as the first-order theory of the class of
expansions $(\mathbb{R}^2,\beta,(P_i)_{i\in\mathbb{N}})$
of the the affine real plane $(\mathbb{R}^2,\beta)$ by unary predicates $P_i\subseteq\mathbb{R}^2$.
Aiello and van Benthem conjecture that the theory is decidable. Expansions of $(\mathbb{R}^2,\beta)$
with \emph{unary} predicates are especially relevant in
investigations related to the geometric structure $(\mathbb{R}^2,\beta)$, since in this 
context unary predicates correspond to 
\emph{regions} of the plane $\mathbb{R}^2$.

In this article we study structures of the type
of $(T,\beta)$, where $T\subseteq\mathbb{R}^n$ and $\beta$ is the
canonical Euclidean betweenness predicate restricted to $T$, see
Section \ref{structures} for the
formal definition.
%We identify a significant collection of canonical structures $(T,\beta)$
%with an undecidable first-order theory of its unary expansion class.
%By the \emph{unary expansion class of} $(T,\beta)$ we mean the class $\{(T,\beta, P) \mid P\subseteq T\}$.
%
Let $E\bigl((T,\beta)\bigr)$
denote the class of expansions $(T,\beta,P)$
of $(T,\beta)$ with a single unary predicate. The class $E\bigl((T,\beta)\bigr)$ is called
the \emph{unary expansion class of $(T,\beta)$}.
We identify a significant collection of canonical structures $(T,\beta)$
with an undecidable first-order theory of $E\bigl((T,\beta)\bigr)$.
Informally, if there exists a flat two-dimensional
region $R\subseteq\mathbb{R}^n$, no matter how small,
such that $T\cap R$ is in a certain sense sufficiently dense with respect to $R$, then the
first-order theory of $E\bigl((T,\beta)\bigr)$ is undecidable.
If the related density conditions are satisfied, we say that \emph{$T$ extends
linearly in $\mathrm{2D}$}, see Section \ref{structures} for the formal definition.
We prove that for any $T\subseteq\mathbb{R}^n$, if $T$ extends linearly in $\mathrm{2D}$,
then the $\mathrm{FO}$-theory of the unary expansion class of $(T,\beta\bigr)$ is $\Sigma^0_1$-hard. We also obtain a partial converse to
this result. We observe that $T$ extending linearly
in $\mathrm{1D}$ (see
Section \ref{structures} for the definition) is not a sufficient condition for undecidability of the $\mathrm{FO}$-theory of $E\bigl((T,\beta)\bigr)$.

In addition, we establish that for all $n\geq 2$,
the first-order theory of the unary expansion class of $(\mathbb{R}^n,\beta)$ is $\Pi_1^1$-hard,
and therefore not even arithmetical.
We thereby refute the conjecture
of Aiello and van Benthem from \cite{vanBenthemAiello:2002}.
The results are ultimately based on tiling arguments. The result establishing $\Pi_1^1$-hardness
relies on the \emph{recurrent tiling problem} of Harel \cite{Harel:1985}---once
again demonstrating the usefulness of Harel's methods.
Our results establish undecidability for a wide range of unary expansion classes of natural geometric
structures $(T,\beta)$.
In addition to $(\mathbb{R}^2,\beta)$, such
structures include for example the rational plane $(\mathbb{Q}^2,\beta)$,
the real unit cube $([0,1]^3,\beta)$, and the plane of
algebraic reals $(\mathbb{A}^2,\beta)$ --- to name a few.
In addition to investigating expansion classes of the type $E\bigl((T,\beta)\bigr)$, we also
study expansion classes with a \emph{restricted} unary predicate. Let $n$ be a positive integer and let $T\subseteq \rsq$.
Let $F\bigl((T,\beta)\bigr)$ denote the class of structures $(T,\beta,P)$, where the set $P$ is a
\emph{finite} subset of $T$. We establish that 
if $T$ extends linearly in $\mathrm{2D}$, then the first-order theory of $F\bigl((T,\beta)\bigr)$
is undecidable.
%
%An alternative reading of this result is that the \emph{weak} universal monadic second-order theory of $(T,\beta)$ is undecidable.
%
We obtain a $\Pi_1^0$-hardness result by an argument based on the
\emph{periodic torus tiling problem} of Gurevich and Koryakov \cite{GurevichKoryakov:1972}.
The torus tiling argument  
can easily be adapted to deal with various different kinds of 
natural restricted expansion classes of geometric structures $(T,\beta)$.
%
\begin{comment}
%
Consider, for example, the class $\mathcal{S}$ of
structures $(\rpm,(P_i)_{i\in\mathbb{N}})$, where
each unary predicate corresponds to a finite union of
open  squares $S\subseteq\rsq$ of some unit size.
This class could be interesting from the point of view of
real-life applications. It is rather immediate that our
torus tiling argument can easily be adapted to deal with the
class $\mathcal{S}$. Therefore the
first-order theory of $\mathcal{S}$ is---unfortunately---undecidable.
Several other expansion classes of $\rp$ can be similarly 
seen to have an undecidable first-order theory.
%
\end{comment}
%
These include classes with a
unary predicate denoting---to name a few examples---a polygon,
a finite union of closed rectangles, and
a semialgebraic set (see \cite{realalgebraicgeometry} for the definition).
Our results could turn out useful in investigations concerning logical
aspects of spatial databases. There is a
canonical correspondence between $(\mathbb{R}^2,\beta)$ and
$(\mathbb{R},0,1,\cdot,+,<)$,
see \cite{Gyssens:1999} for example. 
See the survey \cite{KujpersVandenBussche} for further details on logical
aspects of spatial databases.
The betweenness predicate is also studied in spatial logic \cite{handbookofspatial}. The recent years 
have witnessed a significant increase in the research on spatially motivated logics.
Several interesting systems
with varying motivations have been investigated, see for
example the articles
\cite{vanBenthemAiello:2002, Balbianietal:1997, BalbianiGoranko:2002,
HodkinsonHussain, KoPrWoZa:2010, Nenov:2010, Shere:2010,Tinchev,Venema:1999}.
See also the surveys \cite{AielloPrattHartmannvanBenthem:2007} and \cite{Balbianietal:2007} in the
Handbook of Spatial Logics
\cite{handbookofspatial}, and the Ph.D. thesis \cite{Griffiths:2008}.
Several of the above articles investigate
fragments of first-order
theories by way of modal logics for
affine, projective, and metric geometries. Our results contribute to the 
understanding of spatially motivated first-order languages, and hence they can be useful in the
search for decidable (modal) spatial logics.
In addition to studying issues of decidability, we briefly 
compare the expressivities of universal monadic second-order logic $\forall\mathrm{MSO}$ 
and weak universal monadic second-order logic $\forall\mathrm{WMSO}$.
It is straightforward to observe that in general,
the expressivities of $\forall\mathrm{MSO}$ and $\forall\mathrm{WMSO}$ are incomparable in a rather 
strong sense: $\forall\mathrm{MSO}\not\leq\mathrm{WMSO}$ and $\forall\mathrm{WMSO}\not\leq\mathrm{MSO}$.
Here $\mathrm{MSO}$ and $\mathrm{WMSO}$ denote monadic second-order logic and weak monadic second-order logic,
respectively. The result $\forall\mathrm{WMSO}\not\leq\mathrm{MSO}$ follows from known
results (see \cite{tenCate:2011} for example), and the result $\forall\mathrm{MSO}\not\leq\mathrm{WMSO}$ is
more or less trivial to prove.
While $\forall\mathrm{MSO}$ and $\forall\mathrm{WMSO}$ are
incomparable in general,
the situation changes when we consider
expansions $(\rpm,(R_i)_{i\in I})$ of the structure $\rp$, i.e., structures
embedded in the geometric structure $\rp$. Here $(R_i)_{i\in I}$ is an arbitrary vocabulary and $I$ an
arbitrary related index set. We show that over such structures,
sentences of $\forall\mathrm{WMSO}$ translate into equivalent
sentences of $\forall\mathrm{MSO}$. The result follows immediately from
our proof that the expansion class $F\bigl((\mathbb{R}^n,\beta)\bigr)$
is first-order definable with respect to the class $E\bigl((\mathbb{R}^n,\beta)\bigl)$, see Section \ref{section3}.
The proof is based on the Heine-Borel theorem.
The structure of the current article is as follows. In Section \ref{section2}
we define the central notions
needed in the later sections. In Section \ref{section3} we compare the
expressivities of $\forall\mathrm{MSO}$ and $\forall\mathrm{WMSO}$. 
In Section \ref{main} we show undecidability of the first-order theory of the unary expansion class of
any geometric structure $(T,\beta)$ such that $T$ extends  linearly in $\mathrm{2D}$. In addition, we show
that for, $n\geq 2$, the first-order theory of the unary expansion class of $\rp$ is
not arithmetical. In Section \ref{section5} we
modify the approach in Section \ref{main} and
show undecidability of the $\mathrm{FO}$-theory of any class
$F\bigl((T,\beta)\bigr)$ such that $T$ extends linearly in $\mathrm{2D}$.

For further information about 
Tarski and the facts proved by his school
regarding fragments of ordered affine geometry relevant to the current paper,
see \cite{schw} and \cite{TarskiGivant:1999},
and the papers \cite{prestel1} and \cite{prestel2}.
%
%See also \cite{pambuccian}, which, inter alia, summarizes the results of
%
%\cite{prestel1} and \cite{prestel2}.
%
For a comprehensive survey on the development of the axiomatics of geometries of order, see \cite{pambuccian}, which, among other things, summarizes the results of
\cite{prestel1} and \cite{prestel2}.

This article is an extended version of the conference paper \cite{KMV:2012}. 
\section{Preliminaries}\label{section2}
\subsection{Interpretations}
Let $\sigma$ be a \emph{purely relational vocabulary}, i.e., a vocabulary that does not
contain function symbols or constant symbols.
Let $\tau$ be a vocabulary that does not contain function symbols.
Let $\mathcal{B}$ be a
nonempty class of $\sigma$-structures and  $\mathcal{C}$ a nonempty class of $\tau$-structures.
Assume that there exists a surjective map $F$ from $\mathcal{C}$ onto $\mathcal{B}$ and a
first-order $\tau$-formula $\varphi_{\mi{Dom}}(x)$ in one free variable, $x$,
such that for each structure $\mf{C}\in\mathcal{C}$, there is a bijection $f$ from the domain of $F(\mf{C})$ to the set
$$\{\ u\in \mi{Dom}(\mf{C})\ |\ \mf{C}\models\varphi_{\mi{Dom}}(u)\ \}.$$
Assume, furthermore, that for each relation symbol $R\in\sigma$, there is a first-order $\tau$-formula $\varphi_R(x_1,...,x_{\mi{Ar}(R)})$ such that
we have
$$R^{F(\mf{C})}(u_1,...,u_{\mi{Ar}(R)})\ \Leftrightarrow\ \mf{C}\models\varphi_R\bigl(f(u_1),...,f(u_{\mi{Ar}(R)})\bigr)$$
for every tuple $(u_1,...,u_{\mi{Ar}(R)})\in(\mi{Dom}(F(\mf{C})))^{\mi{Ar}(R)}$. Here $\mi{Ar}(R)$ is the arity of $R$.
We then say that the class $\mathcal{B}$ is \emph{uniformly first-order
interpretable in $\mathcal{C}$}. 

Assume that a class of $\sigma$-structures $\mathcal{B}$ is uniformly first-order interpretable in a class $\mathcal{C}$ of $\tau$-structures.
%Let $\mathcal{P}$ be a set of
%
%unary relation symbols such that $\mathcal{P}\cap (\sigma\cup\tau)\, =\, \emptyset$. 
Define a map $I$ from the set of first-order $\sigma$-formulae to
the set of first-order $\tau$-formulae as follows. 
\begin{enumerate}
%
%\item
%If $P\in\mathcal{P}$, then $I(Px)\, :=\, Px$.
%
\item
If $k\in\mathbb{N}_{\geq 1}$ and $R\in\sigma$ is a $k$-ary relation symbol,
then
$$I(R(x_1,...,x_k))\, :=\, \varphi_{R}(x_1,...,x_k),$$
where $\varphi_{R}(x_1,...,x_k)$ is a first-order formula for $R$ witnessing the
fact that $\mathcal{B}$ is uniformly first-order 
interpretable in\, $\mathcal{C}$.
\item
$I(x=y)\, :=\, x=y$.
\item
$I(\neg\varphi):=\neg I(\varphi)$.
\item
$I(\varphi\wedge\psi)\, :=\, I(\varphi)\wedge I(\psi)$.
\item
$I\bigl(\exists x\, \psi(x)\bigr)\, :=\, \exists x\bigl(\varphi_{\mi{Dom}}(x)\wedge I(\psi(x))\bigr).$
\end{enumerate}
We call the map $I$ a \emph{uniform interpretation of\, $\mathcal{B}$ in $\mathcal{C}$.}
Also, if $\mathcal{A}$ is the class of reducts of structures $\mf{B}\in\mathcal{B}$ to some
vocabulary $\rho\subseteq\sigma$, the function $I$ is called a uniform interpretation of $\mathcal{A}$ in $\mathcal{C}$.
%
%In the case where $\mathcal{P}$ is empty, the map $I$ is a \emph{uniform
%
%interpretation of\, $\mathcal{A}$ in $\mathcal{C}$}.
%

%
\begin{lem}\label{uniforminterpretationlemma}
Let $\rho$ be a purely relational vocabulary and $\tau$ a vocabulary 
not containing function symbols. Let $\mathcal{A}$ be a class of $\rho$-structure
and\, $\mathcal{C}$ a class of $\tau$-structures. 
Let $I$ be a uniform interpretation of $\mathcal{A}$ in $\mathcal{C}$.
Let $\varphi$ be a first-order formula of the vocabulary $\rho$.
Then the following conditions are equivalent.
\begin{enumerate}
\item
There exists a structure $\mf{A}\in\mathcal{A}$ such
that $\mf{A}\models \varphi$.
\item
There exists a structure $\mf{C}\in\mathcal{C}$ such that 
$\mf{C}\models I(\varphi)$.
\end{enumerate}
\end{lem}
\begin{proof}
Straightforward.
\end{proof}

\subsection{Logics and structures}\label{logics and structures}

\emph{Monadic second-order logic}, $\mathrm{MSO}$,
extends first-order logic with
quantification of
relation symbols ranging over subsets of the domain of a model.
In \emph{universal (existential) monadic second-order logic}, $\forall \mathrm{MSO}$ ($\exists \mathrm{MSO}$),
the quantification of monadic relations is restricted to universal (existential)
prenex quantification in the beginning of formulae.
The logics $\forall \mathrm{MSO}$ and $\exists \mathrm{MSO}$ are also known as monadic $\Pi_1^1$ and monadic $\Sigma_1^1$, respectively.
\emph{Weak monadic second-order logic}, $\mathrm{WMSO}$, is a semantic variant of
monadic second-order logic in which the quantified relation symbols range over
finite subsets of the domain of a model.
The weak variants $\forall \mathrm{WMSO}$
and $\exists \mathrm{WMSO}$ of $\forall \mathrm{MSO}$ and $\exists\mathrm{MSO}$ are
defined in the obvious way. For further information on $\mathrm{MSO}$, see for example 
\cite{gure1} and \cite{Libkin}.
Monadic second-order logic can be characterized by a variant of the Ehrenfeucht-Fra\"{i}ss\'{e} game.
We will give a short description of the game here. A more detailed description can be found in \cite{Libkin}.
An $\mathrm{MSO}$ game is played by two players, the \emph{spoiler} and the \emph{duplicator},
on two structures $\mf{A}$ and $\mf{B}$ of the
same purely relational vocabulary $\sigma$.
A round starts by spoiler picking a structure, $\mf{A}$ or $\mf{B}$, and an element or a subset of that
structure. The duplicator responds by choosing an object of the same type from the
other structure. Let $\vec{a}$ and $\vec{b}$ be the elements and $\vec{A}$ and $\vec{B}$ the
subsets chosen in a $k$-round game from the structures $\mf{A}$ and $\mf{B}$, respectively. Then the
duplicator wins the game iff $(\vec{a},\vec{b})$
defines a partial isomorphism from $(\mf{A},\vec{A})$ to $(\mf{B},\vec{B})$.

The $k$-round game characterizes $\mathrm{MSO}[k]$,
the fragment of $\mathrm{MSO}$ up to the quantifier nesting depth $k$.
More formally: If $\mf{A}$ and $\mf{B}$ are two
structures of the same purely relational vocabulary, then the duplicator has a winning
strategy in the $k$-round $\mathrm{MSO}$ game on $\mf{A}$ and $\mf{B}$ iff $\mf{A}$ and $\mf{B}$
agree on all sentences of $\mathrm{MSO}[k]$.

Let $\mathcal{L}$ be any fragment of second-order logic.
The \emph{$\mathcal{L}$-theory} of a structure $\mf{M}$ of a vocabulary $\tau$ is the set of $\tau$-sentences $\varphi$ of
$\mathcal{L}$ such that $\mf{M}\models\varphi$.

Define two binary relations $H,V\subseteq\mathbb{N}^2\times\mathbb{N}^2$ as follows.
\begin{itemize}
\item
$H\ =\ \{\ \bigl((i,j),(i+1,j)\bigr)\ |\ i,j\in\mathbb{N}\ \}$.
\item
$V\ =\ \{\ \bigl((i,j),(i,j+1)\bigr)\ |\ i,j\in\mathbb{N}\ \}$.
\end{itemize}
We let $\mf{G}$ denote the structure $(\mathbb{N}^2,H,V)$, and call
it the \emph{grid}. The relations $H$ and $V$ are
called the \emph{horizontal} and \emph{vertical successor relations} of $\mf{G}$,
respectively. A \emph{supergrid} is a structure of the vocabulary $\{H,V\}$ that has $\mf{G}$ as a
substructure. We denote the class of supergrids by $\mathcal{G}$.
Let $(\mf{G},R)$ be the expansion of $\mf{G}$, where
$R\ =\ \{\ \bigl((0,i),(0,j)\bigr)\in\mathbb{N}^2\times\mathbb{N}^2\ |\ i < j\ \}.$
We denote the structure $(\mf{G},R)$ by $\mf{R}$, and call it the \emph{recurrence grid}.

\begin{figure}
\centering
\includegraphics[scale=1]{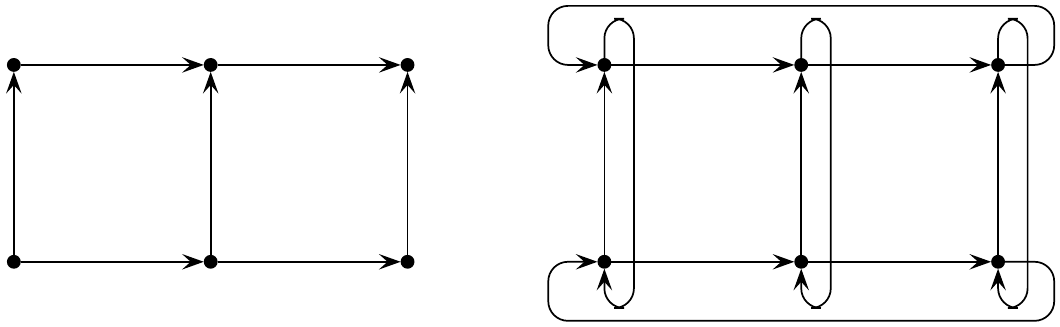}
\caption{The figure shows a $3\times 2$ grid and a $3\times 2$ torus.}
\label{fig:torus}
\end{figure}
Let $m$ and $n$ be positive integers. Define two binary relations $H_{m,n},V_{m,n}\subseteq (m\times n)^2$ as follows.
(Note that we define $m = \{0,...,m-1\}$, and analogously for $n$.)
\begin{itemize}
\item $H_{m,n}=H\upharpoonright (m\times n)^2\cup\{((m-1,i),(0,i))\mid i<n\}$.
\item $V_{m,n}=V\upharpoonright (m\times n)^2\cup\{((i,n-1),(i,0))\mid i<m\}$.
\end{itemize}
We call the structure $(m\times n, H_{m,n},V_{m,n})$ the \emph{$m\times n$ torus} and denote it by  $\mf{T}_{m,n}$.
A torus is essentially a finite grid whose east border wraps back to the
west border and north border back to the south border.
\subsection{Geometric affine betweenness structures}\label{structures}
Let $(\rsq, d)$ be the $n$-dimensional Euclidean space with the canonical metric $d$.
We always assume $n\geq 1$. We define the
ternary Euclidean \emph{betweenness} relation $\beta$ such that $\beta(s,t,u)$ iff 
\[
d(s,u)=d(s,t)+d(t,u).
\]
By $\beta^*$ we
denote the \emph{strict betweenness} relation, i.e., $\beta^*(s,t,u)$ iff  $\beta(s,t,u)$ and $s\not=t\not=u$.
We say that the points $s,t,u\in\mathbb{R}^n$ are \emph{collinear} if the 
disjunction
\[
\beta(s,t,u)\vee \beta(s,u,t)\vee \beta(t,s,u)
\]
holds in $\rp$.
We define the first-order $\{\beta\}$-formula
\[
collinear(x,y,z):=\beta(x,y,z)\vee\beta(x,z,y)\vee\beta(y,x,z).
\]

Below we study geometric betweenness structures of the type $(T,\beta_T)$ where $T\subseteq\mathbb{R}^n$ and $\beta_T=\beta\upharpoonright T$.
Here $\beta\upharpoonright T$ is the restriction of the betweenness predicate $\beta$ of $\mathbb{R}^n$ to the set $T$.
To simplify notation, we usually refer to these structures by $(T,\beta)$.
Let $T\subseteq\rsq$ and let $\beta$ be the corresponding betweenness relation.
We say that $L\subseteq T$ is a \emph{line in $T$} if
the following conditions hold.
\begin{enumerate}
\item
There exist points $s,t\in L$ such that $s\not= t$.
\item
For all $s,t,u\in L$, the points $s,t,u$ are collinear.
\item
Let  $s,t\in L$ be points such that $s\not=t$. For all $u\in T$, if $\beta(s,u,t)$ or $\beta(s,t,u)$, then $u\in L$.
\end{enumerate}
Let $T\subseteq\rsq$ and let $L_1$ and $L_2$ be lines in $T$. We say that $L_1$ and $L_2$
\emph{intersect} if $L_1\not=L_2$ and $L_1\cap L_2\neq\emptyset$. We say that the lines $L_1$ and $L_2$
\emph{intersect in $\rsq$} if $L_1\not=L_2$ and $L_1'\cap L_2'\neq\emptyset$, 
where $L_1',L_2'$ are the lines in $\rsq$ such that $L_1\subseteq L_1'$ and $L_2\subseteq L_2'$.
A subset $S\subseteq\mathbb{R}^n$ is an \emph{$m$-dimensional flat} of $\mathbb{R}^n$, where $0\leq m\leq n$, if there exists a linearly 
independent set of
$m$ vectors $v_1,\dots,v_m\in\mathbb{R}^n$ and a vector $h\in \rsq$ such that $S$ is
the $h$-translated span of the vectors $v_1,\dots,v_m$,
in other words
\[
S=\{u\in\rsq\mid u=h + r_1 v_1 + \dots +r_m v_m,\ r_1,\dots,r_m\in\mathbb{R}\}.
\]
None of the vectors $v_i$ is allowed to be the zero-vector. (This is relevant
in the case where $m=1$.)
A nonempty set $U\subseteq\mathbb{R}^n$ is a \emph{linearly regular $m$-dimensional flat}, where $0\leq m\leq n$, if the following conditions hold.
\begin{enumerate}
\item
There exists an $m$-dimensional flat $S$ such that $U\subseteq S$. 
\item
There does not exist any $(m-1)$-dimensional flat $S$ such that $U\subseteq S$.
\item
$U$ is \emph{linearly complete}, i.e., if $L$ is a line in $U$ and $L'\supseteq L$ the corresponding line in $\mathbb{R}^n$,
and if $r\in L'$ is a point in $L'$ and $\epsilon\in\mathbb{R}_+$ a positive real number, then there exists a point $s\in L$
such that $d(s,r) < \epsilon$. Here $d$ is the canonical metric of $\mathbb{R}^n$.
\item
$U$ is \emph{linearly closed}, i.e., if $L_1$ and $L_2$ are lines in $U$ and $L_1$ and $L_2$ intersect in $\mathbb{R}^n$, 
then the lines $L_1$ and $L_2$ intersect. In other words, there exists a point $s\in U$ such that $s\in L_1\cap L_2$.
\end{enumerate} 
\noindent A set $T\subseteq\mathbb{R}^n$ \emph{extends linearly in $mD$}, where $m\leq n$, if
there exists a linearly regular $m$-dimensional flat $S$, a positive real number $\epsilon\in\mathbb{R}_+$ and a
point $x \in S\cap T$ such that 
$\{\ u\in S\ |\ d(x,u)<\epsilon\ \}\ \subseteq T.$
It is easy show that for example $\mathbb{Q}^2$ extends linearly in ${\mathrm{\mathrm{2D}}}$.

Let $T\subseteq \rsq$, $n\in\mathbb{N}$ and let $\beta$ be the corresponding betweenness relation.  The class of all expansions of $(T,\beta)$ to the vocabulary $\{\beta,P\}$, where $P$ is a unary relation symbol, is called \emph{the unary expansion class of $(T,\beta)$}. By \emph{the unary expansion class of $(T,\beta)$ with a finite predicate}, we mean the class of all expansions of $(T,\beta)$ to the vocabulary $\{\beta,P\}$, where the interpretation of $P$ is a finite set.
\subsection{Tilings}
A function $t:4\longrightarrow\mathbb{N}$ is called a \emph{tile type}.
Define the set
$$\mathrm{TILESYMB}\ :=\ \{\ P_t\ |\ t\text{ is a tile type }\ \}$$
 of unary relation symbols.
The unary relation symbols in the set $\mathrm{TILESYMB}$ are
called \emph{tile symbols}.
The numbers $t(i)$ of a tile symbol $P_t$ are the \emph{colours} of $P_t$.
The number $t(0)$ is the \emph{top colour}, $t(1)$ the \emph{right colour},
$t(2)$ the \emph{bottom colour}, and $t(3)$ the \emph{left colour} of $P_t$.
We then define a lexicographic linear ordering of tile types.
Let $t$ and $s$ be tile types. We define $s < t$, if the tuple $\bigl(s(0),s(1),s(2),s(3)\bigr)$
is situated below the tuple $\bigl(t(0),t(1),t(2),t(3)\bigr)$ with respect to the canonical lexicographic ordering, i.e.,
$s < t$ if there exists some $i\in 4$ such that
\begin{enumerate}
\item
$s(i) < t(i)$, and
\item
$s(j) = t(j)$ for all $j$ such that $j<i$ and $j\in 4$.
\end{enumerate}
If $t$ is a tile type, define $N(t)$ to be the number of tile types $s$ such that $s\leq t$. The function $N$ associates
each tile type with a unique positive integer. 
Let $T$ be a finite nonempty set of tile symbols. We say that a structure $\mf{A}=(A,V,H)$, where $V,H\subseteq A^2$, is
\emph{$T$-tilable}, if there exists an expansion of $\mf{A}$ to
the vocabulary
$$\{H,V\}\cup\{\ P_t\ |\ P_t\in T\ \}$$
such that the following conditions hold.
\begin{enumerate}
\item
Each point of $A$ belongs to the extension of
exactly one symbol $P_t$ in $T$.
\item
If $u H v$ for some points $u,v\in A$,
then the right colour of the tile symbol $P_t$ s.t. $P_{t}(u)$ is the same as the left
colour of the tile symbol $P_{t'}$ such that $P_{t'}(v)$.
\item
If $u V v$ for some points $u,v\in A$,
then the top colour of the tile symbol $P_t$ s.t. $P_{t}(u)$ is the same as the bottom
colour of the tile symbol $P_{t'}$ such that $P_{t'}(v)$.
\end{enumerate}
Let $s$ be a tile type
such that $P_s\in T$. We say that the grid $\mf{G}$ is \emph{$s$-recurrently $T$-tilable}, if
there exists an expansion of $\mf{G}$ to the vocabulary
$$\{H,V\}\cup\{\ P_t\ |\ P_t\in T\ \}$$
such that
the above conditions $(1) - (3)$ hold, and additionally,
there exist infinitely many points $(0,i)\in\mathbb{N}^2$
such that $P_s\bigl((0,i)\bigr)$. 
Intuitively, this means that the tile symbol $P_s$ occurs infinitely many 
times in the leftmost column of the grid $\mf{G}$.
Let $\mathcal{F}$ be the set of finite, nonempty sets $T\subseteq\mathrm{TILESYMB}$, and let
$$\mathcal{H}\ :=\ \{\ (t,T)\ |\ T\in\mathcal{F},\ P_t\in T\ \}.$$
Define the following languages
\begin{align*}
\mathcal{T}\ :=&\ \{\ T\in\mathcal{F}\ \mid\ \mf{G}\text{ is $T$-tilable } \},\\
\mathcal{R}\ :=&\ \{\ (t,T)\in\mathcal{H}\ \mid\ \mf{G}\text{ is $t$-recurrently $T$-tilable } \},\\
\mathcal{S}\ :=&\ \{\ T\in \mathcal{F}\ \mid\
\text{ there is a torus $\mf{D}$ which is $T$-tilable }\}.
\end{align*}
The \emph{tiling problem} is the membership problem of the set $\mathcal{T}$
with the input set $\mathcal{F}$.
The \emph{recurrent tiling problem} is the
membership problem of the set $\mathcal{R}$
with the input set $\mathcal{H}$.
The
\emph{periodic tiling problem} is the
membership problem of $\mathcal{S}$ with the input set $\mathcal{F}$.

\begin{thm}\cite{Berger}
The tiling problem is $\Pi_1^0$-complete.
\end{thm}
\begin{thm}\cite{Harel:1985}
The recurrent tiling problem is $\Sigma_1^1$-complete.
\end{thm}

\begin{thm}\cite{GurevichKoryakov:1972}\label{periodictilingcomplete}
The periodic tiling problem is $\Sigma^0_1$-complete.
\end{thm}

\begin{lem}\label{tilingdefinablelemma}
There is a computable function associating
each input $T$ to the
(periodic) tiling problem with a first-order sentence $\varphi_{T}$ of the
vocabulary $\tau:=\{H,V\}\cup T$ 
such that for all structures $\mf{A}$ of the vocabulary $\{H,V\}$, the structure $\mf{A}$ is $T$-tilable iff
there exists an expansion $\mf{A}^*$ of\, $\mf{A}$ to the
vocabulary $\tau$ such that $\mf{A}^*\models\varphi_{T}$.
\end{lem} 
\begin{proof}
Straightforward.
\end{proof}
\begin{lem}\label{recurrencetilingdefinablelemma}
There is a computable function associating
each input $(t,T)$ of the
recurrent tiling problem with a first-order sentence $\varphi_{(t,T)}$ of the
vocabulary $\tau:=\{H,V,R\}\cup T$ 
such that the grid $\mf{G}$ is $t$-recurrently $T$-tilable iff
there exists an expansion $\mf{R}^*$ of the recurrence grid $\mf{R}$ to the
vocabulary $\tau$ such that $\mf{R}^*\models\varphi_{(t,T)}$.
\end{lem} 
\begin{proof}
Straightforward.
\end{proof}
It is easy to see that the grid $\mf{G}$ is $T$-tilable iff there exists a supergrid $\mf{G'}$ that is $T$-tilable.%
\section{Expressivity of universal $\mathrm{MSO}$ and
weak universal $\mathrm{MSO}$ over affine real structures $(\mathbb{R}^n,\beta)$}\label{section3}
In this section we investigate the expressive powers of $\forall{\mathrm{WMSO}}$
and $\forall\mathrm{MSO}$. While it is rather easy to conclude that the two logics
are incomparable in a rather strong sense (see Proposition \ref{expressivityproposition}),
when attention is limited to structures $(\rpm,(R_i)_{i\in I})$ that expand the affine real structure $\rp$,
sentences of $\forall\mathrm{WMSO}$ translate into equivalent sentences of $\forall\mathrm{MSO}$.
Let $\mathcal{L}$ and $\mathcal{L}'$ be fragments of second-order logic.
We write $\mathcal{L}\leq\mathcal{L}'$, if for every vocabulary $\sigma$,
any class of $\sigma$-structures definable by a $\sigma$-sentence of $\mathcal{L}$ is
also definable by a $\sigma$-sentence of $\mathcal{L}'$.
Let $\tau$ be a vocabulary such that $\beta\not\in\tau$.
The class of all expansions of $\rp$ to the vocabulary $\{\beta\}\cup\tau$ is
called the class of \emph{affine real $\tau$-structures}. Such structures can be
regarded as $\tau$-structures \emph{embedded} in the geometric structure $\rp$.
We say that \emph{$\mathcal{L}\leq\mathcal{L'}$ over $\rp$}, if 
for every vocabulary $\tau$ s.t. $\beta\not\in\tau$,
any subclass definable w.r.t. the class $\mathcal{C}$ of all
affine real $\tau$-structures by a sentence of $\mathcal{L}$ is also
definable w.r.t. $\mathcal{C}$ by a sentence of $\mathcal{L}'$.
We sketch a canonical proof of the following very simple proposition. The result $\forall\mathrm{WMSO}\not\leq\mathrm{MSO}$
follows from already known results (see \cite{tenCate:2011} for example), and the result $\forall\mathrm{MSO}\not\leq\mathrm{WMSO}$ is easy to prove.
\begin{prop}\label{expressivityproposition}
$\forall\mathrm{WMSO}\not\leq\mathrm{MSO}$ and $\forall\mathrm{MSO}\not\leq\mathrm{WMSO}$.
\end{prop}
\begin{proof}[Proof Sketch]
It is easy to observe that $\forall\mathrm{WMSO}\not\leq\mathrm{MSO}$: consider the
sentence $\forall X\exists y\, \neg Xy$. This $\forall\mathrm{WMSO}$ sentence is
true in a model iff the domain of the model is infinite. A straightforward
monadic second-order Ehrenfeucht-Fra\"{i}ss\'{e} game
argument can be used to establish that
infinity is not expressible by any $\mathrm{MSO}$ sentence.
To show that $\forall\mathrm{MSO}\not\leq\mathrm{WMSO}$, consider the
structures $(\mathbb{R},<)$ and $(\mathbb{Q},<)$. The structures can be
separated by a sentence of $\forall\mathrm{MSO}$
stating that every subset bounded
from above has a least
upper bound. To see that the two structures cannot be
separated by any sentence of $\mathrm{WMSO}$, consider the
variant of the $\mathrm{MSO}$ Ehrenfeucht-Fra\"{i}ss\'{e} game
where the players choose \emph{finite sets} in
addition to domain elements. It is easy to establish
that this game characterizes the expressivity of $\mathrm{WMSO}$.
To see that the duplicator has a winning strategy in a game of any finite
length played on the
structures $(\mathbb{R},<)$ and $(\mathbb{Q},<)$, we
devise an extension of the folklore winning strategy in the 
corresponding first-order game. Firstly, the duplicator can obviously 
always pick an element whose  
betweenness configuration corresponds exactly to that of the element picked by the spoiler.
Furthermore, even if the spoiler picks a finite set,
it is easy to see that the duplicator can pick his set such that
each of its elements respect the betweenness
configuration of the set
picked by the spoiler.  
\end{proof}
We then show that $\forall\mathrm{WMSO}\leq\forall\mathrm{MSO}$ and $\mathrm{WMSO}\leq\mathrm{MSO}$ over $(\mathbb{R}^n,\beta)$
for any $n\geq 1$.
\begin{thm}[Heine-Borel]
A set $S\subseteq\rsq$ is closed and bounded iff every open cover of $S$ has a
finite subcover. 
\end{thm}
\begin{thm}
\label{expressivitytheorem}
Let $\mathcal{C}$ be the class of expansions $(\rpm,P)$ of $\rp$ with a unary predicate $P$,
and let $\mathcal{F}\subseteq\mathcal{C}$ be the subclass of $\mathcal{C}$ where $P$ is finite.
The class $\mathcal{F}$ is first-order definable with respect to $\mathcal{C}$. 
\end{thm}
\begin{proof}
We shall first establish that a set $T\subseteq\rsq$ is finite
iff it is closed, bounded and
consists of isolated points of $T$. Recall that an
isolated point $u$ of a set $U\subseteq\rsq$ is a point
such that there exists some open ball $B$ such that $B\cap U =\{u\}$.

Assume $T\subseteq\rsq$ is finite. Since $T$ is finite, we can find a minimum 
distance between points in the set $T$.
Therefore it is clear that each point $t$ in $T$ belongs to some open ball $B$
such that $B\cap T = \{t\}$, and hence $T$
consists of isolated points.
Similarly, since $T$ is finite, each point $b$ in the complement of $T$
has some minimum distance to the points of $T$, and therefore $b$
belongs to some open ball $B\subseteq\rsq\setminus T$.
Hence the set $T$ is the complement of the union of
open balls $B$ such that $B\subseteq\rsq\setminus T$, and therefore $T$ is closed.
Finally, since $T$ is finite, we can find a
maximum distance between the points in $T$,
and therefore $T$ is bounded.
Assume then that $T\subseteq\rsq$ is closed, bounded and consists of isolated points of $T$.
Since $T$ consists of isolated points, it has an open cover $\mathcal{C}\subseteq\mathrm{Pow}(\rsq)$ such
that each set in $\mathcal{C}$ contains exactly one point $t\in T$. The set $\mathcal{C}$ is an open
cover of $T$, and by the Heine-Borel theorem, there exists a
finite subcover $\mathcal{D}\subseteq\mathcal{C}$ of the set $T$.
Since $\mathcal{D}$ is finite and each set in $\mathcal{D}$ contains exactly one point of $T$, 
the set $T$ must also be finite.
We then conclude the proof by establishing that there exists a first-order formula $\varphi(P)$
stating that the unary predicate $P$ is closed, bounded and consists of isolated points. We will first define a formula $\mi{parallel}(x,y,t,k)$ stating that the
lines defined by $x,y$ and $t,k$ are parallel in $\rp$. We define
\begin{align*}
&\mi{parallel}(x,y,t,k):= x\neq y\wedge t\neq k \wedge \Big((\mi{collinear}(x,y,t)\wedge \mi{collinear}(x,y,k))\\
&\quad\vee \big(\neg\exists z(\mi{collinear}(x,y,z)\wedge\mi{collinear}(t,k,z))\\
&\quad\quad\wedge \exists z_1 z_2(x\neq z_1\wedge\mi{collinear}(x,y,z_1)\wedge\mi{collinear}(x,t,z_2)\wedge\mi{collinear}(z_1,z_2,k))\, \big)\Big).
\end{align*}
We will then define first-order $\{\beta\}$-formulae $\mi{basis}_k(x_0,\dots,x_k)$ and $\mi{flat}_k(x_0,\dots,x_k,z)$ using simultaneous recursion.
The first formula states that the vectors corresponding to the pairs $(x_0,x_i)$, $1\leq i\leq k$, form a basis of a $k$-dimensional flat.
The second formula states that the point $z$ is in the span of the basis defined by the vectors $(x_0,x_i)$, the origin being $x_0$.
First define
$\mi{basis}_0(x_0):= x_0=x_0$ and
$\mi{flat}_0(x_0,z):= x_0=z$.
Then define $\mi{flat}_k$ and $\mi{basis_k}$ recursively in the following way.
\begin{align*}
&\mi{basis}_k(x_0,\dots,x_k):=\mi{basis}_{k-1}(x_0,\dots,x_{k-1})\wedge \neg \mi{flat}_{k-1}(x_0,\dots,x_{k-1},x_k),\\
&\mi{flat}_k(x_0,\dots,x_k,z):=\mi{basis}_k(x_0,\dots,x_k)\\
&\quad\wedge \exists y_0\dots \exists y_k\Big(y_0=x_0 \wedge y_k=z
\wedge\bigwedge_{i\, \leq\, k-1} \big(y_i=y_{i+1}\vee\mi{parallel}(x_0,x_{i+1},y_i,y_{i+1})\big)\Big).
\end{align*}
We then define a first-order $\{\beta,P\}$-formula $\mathit{sepr}(x,P)$ 
asserting that $x$ belongs to an open ball $B$ such that each point in $B\setminus\{x\}$
belongs to the complement of $P$.
The idea is to state that there exist $n+1$ points $x_0,\dots,x_n$
that form an \emph{$n$-dimensional triangle} around $x$, and every point contained in the triangle (with $x$ being a
possible exception) belongs to the complement of $P$. Every open ball in $\rsq$ is contained in some $n$-dimensional triangle in $\rsq$ and vice versa. We
will recursively define first-order formulae $\mi{opentriangle}_k(x_0,\dots,x_k,z)$ stating that $z$ is properly inside a $k$-dimensional triangle
defined by $x_0,\dots,x_k$, see Figure \ref{fig:opentriangle}. First define $\mi{opentriangle}_1(x_0,x_1,z):=\beta^*(x_0,z,x_1)$, and then define
\begin{align*}
&\mi{opentriangle}_{k}(x_0,\dots,x_k,z):=\mi{basis}_k(x_0,\dots,x_k)\\
&\quad\wedge \exists y\big(\mi{opentriangle}_{k-1}(x_0,\dots,x_{k-1},y)\wedge \beta^*(y,z,x_k)\big).
\end{align*}
\begin{figure}
\centering
\includegraphics[scale=2]{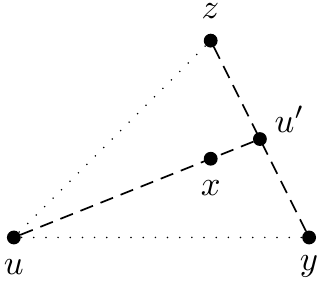}
\caption{The figure shows how the formula $\mi{opentriangle}_2(y,z,u,x)$ is interpreted.}
\label{fig:opentriangle}
\end{figure}We are now ready to define $\mathit{sepr}(x,P)$. Let
\begin{multline*}
\mathit{sepr}(x,P) := \exists x_0\dots \exists x_n\Bigl(\mi{opentriangle}_n(x_0,\dots, x_n,x)\\
\wedge \forall y \big((\mi{opentriangle}_n(x_0,\dots,x_n,y)\wedge y\neq x)\rightarrow \neg Py\big)\Bigr).
\end{multline*}
Now, the sentence
$\varphi_1:=\forall x\bigl(\neg Px\ \rightarrow\ \mathit{sepr}(x,P)\bigr)$
states that each point in the complement of $P$ is contained in an open ball $B\subseteq\rsq\setminus P$.
The sentence therefore states that the complement of $P$ is a union of open balls.
Since the set of unions of open balls is exactly the same as the set of open sets, the sentence states that $P$ is closed.
The sentence
$\varphi_2:=\forall x\bigl(Px\ \rightarrow\ \mathit{sepr}(x,P)\bigr)$
clearly states that $P$ consists of isolated points.
Finally, in order to state that $P$ is bounded, we define a formula asserting that
there exist points $x_0,\dots,x_n$ that
form an n-dimensional triangle around $P$.
\begin{multline*}
\varphi_3:=\exists x_0\dots \exists x_n\Bigl(\mi{basis}_n(x_0,\dots,x_n)
\wedge \forall y\bigl(Py\rightarrow \mi{opentriangle}_n(x_0,\dots,x_n,y)\bigr)\Bigr)
\end{multline*} 
The conjunction $\varphi_1\wedge\varphi_2\wedge\varphi_3$
states that $P$ is finite.
\end{proof}
\begin{cor}
Limit attention to expansions of $\rp$.
Sentences of\, $\forall\mathrm{WMSO}$ translate into
equivalent sentences of $\forall\mathrm{MSO}$,
and sentences of $\mathrm{WMSO}$ into equivalent sentences of\, $\mathrm{MSO}$. 
\end{cor}

%%%%%%%%%%%%%%%%%%%%%%%%%%%%%%%%%%%%%%%%%%%%%%%%%%%%%%%%%%%%%%%%%%%%%%%%%%%%%%%%%%%%%%%%%%%%%%%%%%%%%%%%%%%%%%%%%%%%%%%%%%%%

\section{Undecidable theories of geometric structures with an affine betweenness relation}\label{main}

In this section we establish undecidability of the first-order
theory of the unary expansion class
$$\{\  (T,\beta,P)\ |\ P\subseteq T\ \}$$
of any geometric structure $(T,\beta)$ that extends linearly in $\mathrm{2D}$.
We also show that the first-order theories of
unary expansion classes of structures $(\rsq,\beta)$ with $n\geq 2$ are highly undecidable.
More precisely, we show that the theories of
classes based on structures extending linearly in $\mathrm{2D}$ are $\Sigma_1^0$-hard,
while the theories of classes based on structures $(\rsq,\beta)$ with $n\geq 2$ are $\Pi_1^1$-hard---and therefore not even
 arithmetical. 
We establish the results by a reduction from the (recurrent) tiling problem to the problem of deciding
whether a $\{\beta,P\}$-sentence is \emph{satisfied} in some expansion $(T,\beta,P)$ of $(T,\beta)$ (respectively, in some expansion $(\rsq,\beta,P)$ of $\rp$).
The argument is based on interpreting supergrids in corresponding $\{\beta\}$-structures.
%
%
%
%%%%%%%%%%%%%%%%%%%%%%%%%%%%%%%%%%%%%%%%%%%%%%%%%%%%%%%%%%%%%%%%%%%%%%%%%%%%%%%%%%%%%%%%%%%%%%%%%%%%%%%%%%%%%%%%%%%%%%%%%%%%%%%%%%

\subsection{Lines and sequences}
Let $T\subseteq \rsq$. Let $L$ be a line in $T$.
Any nonempty subset $Q$ of $L$ is called a
\emph{sequence} in $T$.
Let $E\subseteq T$ and $s,t\in T$. If $s\not=t$ and if $u\in E$ for all points $u\in T$ such that $\beta^*(s,u,t)$, we say that the 
points $s$ and $t$ are \emph{linearly $E$-connected} (in $(T,\beta)$). 
If there exists a point $v\in T\setminus E$ such that $\beta^*(s,v,t)$,
we say that $s$ and $t$ are \emph{linearly disconnected with respect to $E$} (in $(T,\beta)$).
\begin{defi}\label{discretelyspaceddfn}
Let $Q$ be a sequence in $T\subseteq \rsq$. Suppose that for each $s,t\in Q$ such that $s\not=t$, there exists a point $u\in T\setminus\{s\}$ such that 
\begin{enumerate}
\item
$\beta(s,u,t)$ and
\item
$\forall r\, \in\, T\ \bigl(\, \beta^*(s,r,u)\rightarrow r\not\in Q\, \bigr)$, i.e., the points $s$ and $u$ are linearly $(T\setminus Q)$-connected.
\end{enumerate}
Then we call $Q$ a $\emph{discretely spaced sequence in T}$.
\end{defi}
\begin{defi}\label{discretelyinfinitedfn}
Let $Q$ be a discretely spaced sequence in $T\subseteq \rsq$. Assume that there exists a point $s\in Q$ such that for each point $u\in Q$, there exists a point
 $v\in Q\setminus\{u\}$ such that $\beta(s,u,v)$.
Then we call the sequence $Q$ a \emph{discretely infinite sequence in $T$}. The point $s$ is
called a \emph{base point} of $Q$. 
\end{defi}
\begin{defi}\label{sequencewithzero}
Let $Q$ be a sequence in $T\subseteq \rsq$. Let $s\in Q$ be a
point such that there do not exist
points $u,v\in Q\setminus\{s\}$ such that $\beta(u,s,v)$.
Then we call $Q$ a \emph{sequence in $T$ with a zero}.
The point $s$ is a \emph{zero-point of $Q$}.
Notice that $Q$ may have up to two
zero-points.
\end{defi}
It is easy to see that a discretely infinite sequence has at most one zero-point.
\begin{defi}\label{omegalikedfn}
Let $Q$ be a discretely infinite sequence in $T\subseteq\rsq$ with a zero.
Assume that for each $r\in T$ such that there exist points $s,u\in Q\setminus\{r\}$
with $\beta(s,r,u)$, there also exist points $s',u'\in Q\setminus\{r\}$ such that
\begin{enumerate}
\item
$\beta(s',r,u')$ and
\item
$\forall v\, \in\, T\setminus\{r\}\ \bigl(\, \beta^*(s',v,u')\rightarrow v\not\in Q\, \bigr)$.
\end{enumerate}
Then we call $Q$ an \emph{$\omega$-like sequence in $T$} (cf. Lemma \ref{nonstandardmodelsexcluded}).
\end{defi}
%
%
%It is easy to see that, for example, the natural numbers embedded into the $x$-axis %of $\rp$ is an $\omega$-like sequence.
%
%
%
%
\begin{lem}\label{omegalikesequencedefinability}
Let $P$ be a unary relation symbol. There is a first-order sentence $\varphi_{\omega}(P)$ of the vocabulary $\{\beta,P\}$ such that
for every $T\subseteq \rsq$ and for every expansion $(T,\beta,P)$ of $(T,\beta)$, we have $(T,\beta,P)\models\varphi_{\omega}(P)$ if and only if the
interpretation of $P$ is an $\omega$-like sequence in $T$.  
\end{lem}
\begin{proof} 
Define
\[
\mi{sequence}(P):=\exists x\, Px\,
\wedge\, \forall x\forall y\forall z\,
\bigl(Px\wedge Py \wedge Pz\ \rightarrow\ \mi{collinear}(x,y,z)\bigr).
\]
The formula $sequence(P)$ states that $P$ is a sequence. By inspection of Definition \ref{discretelyspaceddfn}, it is
easy to see that there is a first-order formula $\psi$ such that the conjunction $sequence(P)\wedge\psi$ states
that $P$ is a discretely spaced sequence. Continuing this trend, it is straightforward to observe
that Definitions \ref{discretelyinfinitedfn}, \ref{sequencewithzero} and \ref{omegalikedfn} specify
first-order properties, and therefore there exists a first-order formula $\varphi_{\omega}(P)$ stating
that $P$ is an $\omega$-like sequence.
\end{proof}
\begin{defi}\label{successordefinition}
Let $P$ be a sequence in $T\subseteq \rsq$ and $s,t\in P$. The points $s,t$ are
called \emph{adjacent} with respect to $P$, if the
points are linearly $(T\setminus P)$-connected.
Let $E\subseteq P\times P$ be the set of pairs $(u,v)$ such that 
\begin{enumerate}
\item
$u$ and $v$ are adjacent with respect to $P$, and 
\item
$\beta(z,u,v)$ for some zero-point $z$ of $P$.
\end{enumerate}
We call $E$ the \emph{successor relation} of $P$.
\end{defi}
We let $\mi{succ}$ denote the successor relation of $\mathbb{N}$, i.e.,
$\mi{succ}:=\{\ (i,j)\in\mathbb{N}\times\mathbb{N}\ |\ i+1 = j\ \}.$

\begin{lem}\label{nonstandardmodelsexcluded}
Let $P$ be an $\omega$-like sequence in $T\subseteq\rsq$ and $E$ the successor relation of $P$.
There is an embedding from $(\mathbb{N},\mi{succ})$ into $(P,E)$ such that $0\in\mathbb{N}$ maps to the zero-point of $P$.
If $T=\rsq$, then $(\mathbb{N},\mi{succ})$ is isomorphic to $(P,E)$.
\end{lem}
\begin{proof}
We denote by $i_0$ the unique zero-point of $P$.
Since $P$ is a discretely infinite sequence, it has a base point.
Clearly $i_0$ has to be the only base point of $P$.
It is straightforward to establish that since $P$ is an $\omega$-like sequence with the base point $i_0$, there exists a sequence
 $(a_i)_{i\in \mathbb{N}}$ of points $a_i\in P$ such that $i_0=a_0$ and $a_{i+1}$ is the unique $E$-successor of $a_{i}$ for all $i\in \mathbb{N}$.
Define the function $h:\mathbb{N}\rightarrow P$ such that $h(i) = a_i$ for all $i\in\mathbb{N}$.
It is easy to see that $h$ is an embedding of $(\mathbb{N},\mi{succ})$ into $(P,E)$.
Assume then that $T=\rsq$. We shall show that the function $h:\mathbb{N}\longrightarrow P$ is a surjection.
Let $d$ denote the canonical metric of $\mathbb{R}$, and let $d_R$ be the
restriction of the canonical metric of $\rsq$ to the line $R$ in $\rsq$ such that $P\subseteq R$.
Let $g:\mathbb{R}\longrightarrow R$ be the isometry from $(\mathbb{R},d)$ to $(R,d_R)$ such
that $g(0) = i_0 = h(0)$ and such
that for all $r\in\mi{ran}(h)$, we have $\beta\bigl(i_0,g(1),r\bigr)$ or $\beta\bigl(i_0,r,g(1)\bigr)$.
Let $(R,\leq^R)$ be the
structure, where
$$\leq^{R}\ =\ \{\ (u,v)\in R\times R\ |\ g^{-1}(u)\, \leq^{\mathbb{R}}\, g^{-1}(v)\ \}.$$
If $\mi{ran}(h)$ is not bounded from above w.r.t. $\leq^R$, then $h$
must be a surjection. Therefore
assume that $\mi{ran}(h)$ is bounded above. By the Dedekind completeness of the reals, there
exists a least upper bound $s\in R$ of $\mi{ran}(h)$ w.r.t. $\leq^R$.
Notice that since $h$ is an 
embedding of $(\mathbb{N},\mi{succ})$ into $(P,E)$, we have $s\not\in\mi{ran}(h)$.
Due to the definition of $E$, it is sufficient to show that $\{\, t\in P\ |\ s\leq^R t\ \} = \emptyset$ in order to conclude
that $h$ maps onto $P$.
Assume that the least upper bound $s$ belongs to the set $P$. Since $P$ is a discretely spaced sequence,
there is a point $u\in\mathbb{R}^n\setminus\{s\}$ such that $\beta(s,u,i_0)$ and
$$\forall r\in\mathbb{R}^n
\bigl(\beta^*(s,r,u)
\rightarrow r\not\in P\bigl).$$
Now $u<^R s$ and the points $u$ and $s$ are linearly $(\mathbb{R}^n\setminus P)$-connected,
implying that $s$ cannot be the
least upper bound of $\mi{ran}(h)$. This is a
contradiction. Therefore $s\not\in P$.
Assume, ad absurdum, that there exists a point $t\in P$
such that $\beta(i_0,s,t)$.
Now, since $P$ is an $\omega$-like sequence, there
exists points $u',v'\in P\setminus\{s\}$
such that
$\beta(u',s,v')$ and
$$\forall r\in \rsq\bigl(\beta^*(u',r,v')\rightarrow r\not\in P\bigr).$$
We have $\beta(s,u',i_0)$ or $\beta(s,v',i_0)$. Assume, by symmetry, that $\beta(s,u',i_0)$.
Now $u'<^R s$, and the points $u'$ and $s$ are
linearly $(\rsq\setminus P)$-connected.
Hence, since  $s\not\in \mi{ran}(h)$, we conclude that
$s$ is not the least upper bound of $\mi{ran}(h)$.
This is a contradiction.
\end{proof}
%
%
%%%%%%%%%%%%%%%%%%%%%%%%%%%%%%%%%%%%%%%%%%%%%%%%%%%%%%%%%%%%%%%%%%%%%%%%%%%%
%
\subsection{Geometric structures with an undecidable unary expansion class}
\begin{figure}
\centering
\includegraphics[scale=1.1]{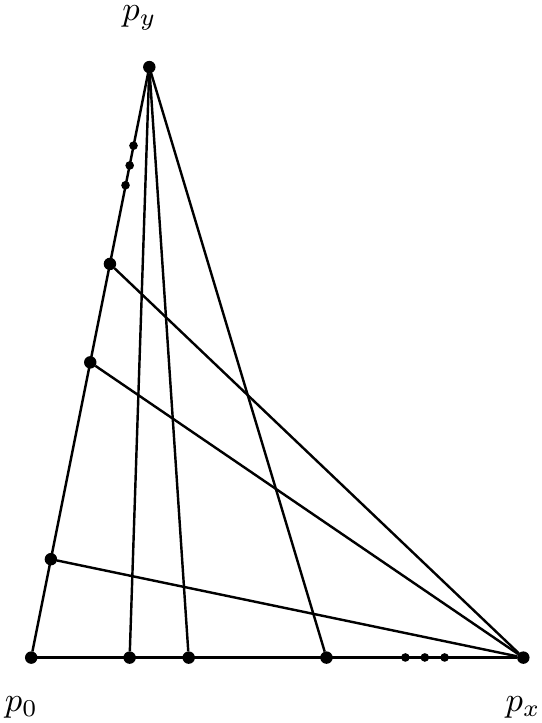}
\caption{The figure illustrates how a grid is interpreted in a Cartesian frame.
The intersection points of the solid lines correspond to domain points of the grid.
See also figure \ref{fig:triangle_tiled}.}
\label{fig:triangle}
\end{figure}
\begin{figure}
\centering
\includegraphics[scale=1.1]{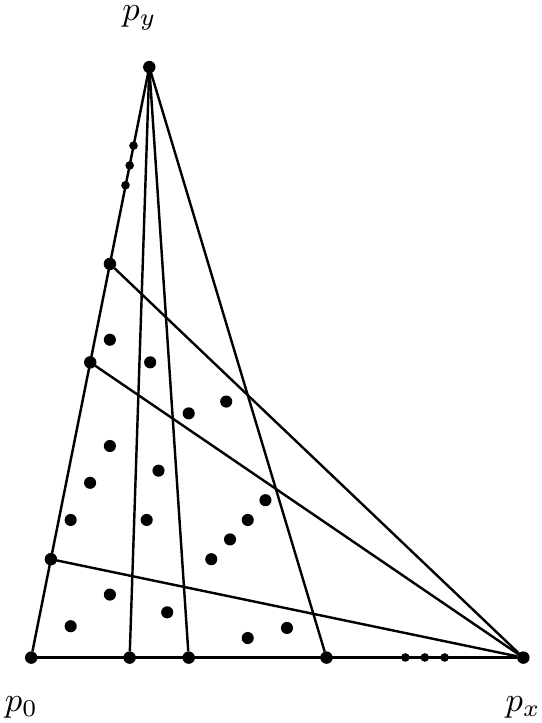}
\caption{The figure shows how the tile symbols of a labelled grid are interpreted in a Cartesian frame.
Each tile symbol $P_t$ is associated with the natural number $N(t)$ (see the Preliminaries section).
The number of dots $n$ on the southwest to northeast diagonal of the rectangle to the northeast of an intersection point $x$ corresponds to the
tile symbol associated with $x$. The point $x$ is associated with $P_t$ iff $n=N(t)$.}
\label{fig:triangle_tiled}
\end{figure}
Let $Q$ be an $\omega$-like sequence in $T\subseteq \rsq$ and let $q_0$ be the unique zero-point of $Q$. Assume there exists a point $q_e\in T\setminus Q$
such that $\beta(q_0,q,q_e)$ holds for all $q\in Q$. We call $Q\cup\{q_e\}$ \emph{an $\omega$-like sequence with an endpoint in $T$}. The point $q_e$ is
the \emph{endpoint of $Q\cup\{q_e\}$}. Notice that the endpoint $q_e$ is the only point $x$ in $Q\cup\{q_e\}$ such that the following conditions hold.
\begin{enumerate}
\item
There do not exist points $s,t\in Q\cup\{q_e\}$ such that $\beta^*(s,x,t)$.
\item
$\forall y z \in Q\cup\{q_e\}\Bigl(\, \beta^*(x,y,z)\rightarrow \exists v\in Q\cup\{q_e\}\beta^*(x,v,y)\Bigr)$.
\end{enumerate}
%
%
%  
%\begin{defi}
%
%
%
%Let $P$ and $Q$ be $\omega$-like sequences with an endpoint in $T\subseteq \rsq$.
%
%Let $p_e$ and $q_e$ be the endpoints of $P$ and $Q$, respectively. Assume that the following conditions hold.
%\begin{enumerate}
%\item There exists a point $z\in P\cap Q$ such that $z$ is the zero-point of both $P\setminus\{p_e\}$ and $Q\setminus\{q_e\}$.
%
%\item There exists lines $L_P$ and $L_Q$ in $T$ such that $L_P\neq L_Q$, $P\subseteq L_P$ and $Q\subseteq L_Q$.
% 
%\item For each point $p\in P\setminus\{p_e\}$ and $q\in Q\setminus\{q_e\}$, the %unique lines $L_p$ and $L_q$ in $T$ such
%
%that $p,q_e\in L_p$ and $q,p_e\in L_q$ intersect.
%\end{enumerate} 
%We call the structure $(T,\beta,P,Q)$ a
%
%\emph{Cartesian frame}.
%
%\end{defi}
%
\begin{defi}
Let $P\subseteq T\subseteq \rsq$, and let $p_0,p_x,p_y\in P$. We call the structure
$$\mf{C}=(T,\beta,P,p_0,p_x,p_y)$$
a \emph{Cartesian frame} with domain $T$, if the following conditions hold.
\begin{enumerate}
\item The points $p_0$, $p_x$ and $p_y$ are not collinear, i.e., $\mi{collinear}(p_0,p_x,p_y)$ does not hold
in the structure $\mf{C}$.
\item
The set 
$$P_x=\{\ u\in P\mid \mi{collinear}(p_0,u,p_x) \text{ holds in }\mf{C}\ \}$$
is an $\omega$-like sequence with an
endpoint in $T$. The point $p_x$ is the endpoint of $P_x$.
\item
The set
$$P_y=\{\ u\in P\mid \mi{collinear}(p_0,u,p_y)\text{ holds in }\mf{C}\ \}$$
is an $\omega$-like sequence with an
endpoint in $T$. The point $p_y$ is the endpoint of $P_y$.
\item The point $p_0$ is the zero-point of both $P_x\setminus\{p_x\}$ and $P_y\setminus\{p_y\}$.
\item For each point $p\in P_x\setminus\{p_x\}$ and $q\in P_y\setminus\{p_
y\}$, the unique lines $L_p$ and $L_q$ in $T$ such
that $p,p_y\in L_p$ and $q,p_x\in L_q$, intersect. In other words, there exists a point $u\in T$
that lies on both lines $L_p$ and $L_q$.
\end{enumerate}
\end{defi}
\begin{defi}\label{intersectionpointdefinition}
Let $\mf{C}=(T,\beta,P,p_0,p_x,p_y)$ be a Cartesian frame. 
Let $p\in P_x\setminus\{p_x\}$ and $q\in P_y\setminus\{p_y\}$
be points and $L_p$ and $L_q$ the lines in $T$ such that $p,p_y\in L_p$ and $q,p_x\in L_q$.
The point $u\in T$ that lies on both lines $L_p$ and $L_q$ is called---rather suggestively---the
\emph{intersection point of $\mf{C}$ corresponding to the pair $(p,q)$}. A point $u\in T$ is called an \emph{intersection point} of the Cartesian
 frame $\mf{C}$, if it is an intersection point of $\mf{C}$ corresponding some pair $(p,q)$, where $p\in P_x\setminus\{p_x\}$ and $q\in
 P_y\setminus\{p_y\}$.
\end{defi}
\begin{defi}\label{diagonalsuccessordefinition}
Let $\mf{C}=(T,\beta,P,p_0,p_x,p_y)$ be a Cartesian frame.
Recall Definition \ref{successordefinition}.
Let $E_x$ be the successor relation of the $\omega$-like sequence $P_x\setminus\{p_x\}$ and $E_y$ the
successor relation of $P_y\setminus\{p_y\}$. Let $p,p',q,q'$ be points
such that $(p,p')\in E_x$ and $(q,q')\in E_y$.
Let $u$ be the intersection point of $\mf{C}$ corresponding to $(p,q)$ and $v$ the
intersection point of $\mf{C}$ corresponding to $(p',q')$.
We say that $v$ is the \emph{diagonal successor} of $u$ in $\mf{C}$.
\end{defi}
\begin{defi}
Recall the function $N$ that associates each tile type $t$ with the unique positive integer $N(t)$ (see the 
Preliminaries section).
Let $\mf{C}=(T,\beta,P,p_0,p_x,p_y)$ be a Cartesian frame and let $S\not=\emptyset$ be a finite set of tile symbols.
We call $\mf{C}$ an \emph{$S$-labelled Cartesian frame} if the number of points in $P$ strictly between
any intersection point $u$ of $\mf{C}$ and its diagonal successor $v$ is in the set $\{N(t)\mid P_t\in S\}$.
If $T\subseteq\rsq$ and $S\not=\emptyset$ is a finite set of tile symbols, we let
$\mathcal{C}(T,S)$ denote the class of exactly all $S$-labelled Cartesian frames with domain $T$.
\end{defi}
%
%\begin{defi}
%
%Let $\mathcal{C}$ be a class of $S$-labelled Cartesian frames with domain $T$ and let $\mf{C}\in\mathcal{C}$. We say that %\emph{the class $\mathcal{C}$ is complete with respect to $\mf{C}$} if for every function $f$ from the set of intersection points of %$\mf{C}$ to $S$ there exists a Cartesian frame $\mf{C}_f\in \mathcal{C}$ such that the following conditions hold.
%\begin{enumerate}
%\item The intersection points of $\mf{C}_f$ are the same as the intersection points of $\mf{C}$.
%\item For each intersection point $u$ of $\mf{C}_f$ and its diagonal successor $v$ the number of points in $P$ strictly between $u$ %and $v$ in $\mf{C}_f$ is $f(u)$.
%\end{enumerate}  
%We say that the class $\mathcal{C}$ is complete if it is complete with respect to some Cartesian frame $\mf{C}\in\mathcal{C}$.
%
%
%\end{defi}
%
%
\begin{lem}\label{boxedcartesianpicturesdefinablelemma}
Let $T\subseteq\rsq$, $n\geq 2$, and let $\mathcal{C}$ be the class of all expansions $(T,\beta,P,p_0,p_x,p_y)$ of $(T,\beta)$ by
a unary relation $P$ and constants $p_0,p_x,p_y$.
There is a computable function associating each input $S$ to the tiling problem
with a first-order $\{\beta, P,p_0,p_x,p_y\}$-sentence $\varphi_\mi{Cf}^S$ that
defines the class of $S$-labelled Cartesian frames with the domain $T$
with respect to the class $\mathcal{C}$.
\end{lem}
\begin{proof}
Straightforward by virtue of Lemma \ref{omegalikesequencedefinability}.
\end{proof}
%
%
%\begin{defi}\label{def:labelling}
%
%Let $\tau$ be the vocabulary $\{V,H\}$ of supergrids.
%
%Let $S$ be a finite nonempty set of tile symbols and $\mf{A}$ a structure
%
%of the vocabulary $\tau\cup \{\ P_t\ |\ P_t\in S\ \}$.
%
%We call the structure $\mf{A}$ an \emph{$S$-labelled $\tau$-structure} if each point in the
%domain of $\mf{A}$
%is in the extension of exactly one tile symbol in $S$.
%
%We note that $S$-labelled $\tau$-structures should not be confused with $S$-labelled Cartesian frames.
%
%\end{defi}
%

%
Recall that we let $\mf{G}$ denote the grid. Let $S\not=\emptyset$
be a finite set of tile symbols. We let $\mf{G}_S$ denote the class of structures $\mf{A}$ that satisfy the following conditions.
\begin{enumerate}
\item The structure $\mf{A}$ is an expansion of the grid $\mf{G}$ to the vocabulary $\{V,H\}\cup S$.
\item Each point in the domain of $\mf{A}$ is in the extension of exactly one predicate symbol $P_t\in S$.
\end{enumerate} 
Structures in the class $\mf{G}_S$ are called \emph{$S$-labelled grids}.
Similarly, we let $\mathcal{G}_S$ denote the class of structures $\mf{A}$ that satisfy the following conditions.
\begin{enumerate}
\item The structure $\mf{A}$ is an expansion of some supergrid to the vocabulary $\{V,H\}\cup S$.
\item Each point in the domain of $\mf{A}$ is in the extension of exactly one predicate symbol $P_t\in S$.
\end{enumerate} 
Structures in the class $\mathcal{G}_S$ are called
\emph{$S$-labelled supergrids}.
%By $\mathcal{R}_S$ we denote the class of structures %$\mf{A}$ that satisfy the following conditions.
%\begin{enumerate}
%\item The structure $\mf{A}$ is an expansion of the recurrence grid $\mf{R}$ to the %vocabulary $\{V,H,R\}\cup S$.
%\item Each point in the domain of $\mf{A}$ is in the extension of exactly one %predicate symbol $P_t\in S$.
%\end{enumerate} 
%Structures in the class $\mathcal{R}_S$ are called 
The class of \emph{$S$-labelled recurrence grids} is defined in the obvious way.
\begin{lem}\label{gridinterpretablelemma}                                                                                                                                                                 
Let $T\subseteq\mathbb{R}^n$  be a set that extends linearly in $\mathrm{2D}$.
There is a computable function $I$ such that for each finite set of tile symbols $S\not=\emptyset$,
there exists some class
$\mathcal{G}(T,S)\supseteq \mf{G}_S$ of $S$-labelled supergrids such that 
the function $I$ is a uniform interpretation of
$\mathcal{G}(T,S)$ in $\mathcal{C}(T,S)$.
\end{lem}
\begin{proof}
Let $S\not=\emptyset$ be a finite set of tile symbols.
Let $\mf{C} = (T,\beta,P,p_0,p_x,p_y)$ be an $S$-labelled Cartesian frame.
%Let $p_e\in P$ and $q_e\in Q$ be the endpoints of $P$ an $Q$, respectively.
%Let $S$ be a set of tile symbols.
%
We shall show how an $S$-labelled supergrid $\mf{G}_{\mf{C}}$ is interpreted in $\mf{C}$.
Figures \ref{fig:triangle} and \ref{fig:triangle_tiled} illustrate the related constructions.
The domain of the interpretation of $\mf{G}_{\mf{C}}$ in $\mf{C}$ will be the set of intersection
points of two sets of lines defined as follows.
The first set of lines 
is formed by connecting the point $p_y$ to each of the points $u$ in the set
$$\{x\in P\mid \beta(p_0,x,p_x)\}\setminus\{p_x\},$$
i.e., each line in the set corresponds to a pair $(u,p_y)$ for some $u\not=p_x$ such that $\beta(p_0,u,p_x)$.
Similarly, the second set of lines is formed by connecting the point $p_x$ to each of the points
in the set
$$\{x\in P\mid \beta(p_0,x,p_y)\}\setminus\{p_y\}.$$
%

%
%
%
%
%
%
%First let us define the following formula which states in $\mf{C}$ that $x$ is the endpoint of $P$.
%
%
%
%\[
%
%\mi{end}_P(P,Q,x)\ :=\ Px \wedge \neg Qx \wedge \neg \exists y \exists z\bigr(Py\wedge Pz\wedge \beta^*(y,x,z)\bigl)
%
%
%\]
%
%
%
%In the following, we let atomic expressions of the type $x\not= p_e$ and $\beta^*(x,y,q_e)$ abbreviate
%
%corresponding first-order formulae $\exists z\bigl(\mi{end}_P(P,Q,z)\wedge x\not= z\bigr)$ and $\exists z\bigl(\mi{end}_Q(Q,P,z)\wedge \beta^*(x,y,z)\bigr)$ of
%
%the vocabulary $\{\beta,P,Q\}$ of $\mf{C}$.
%
We then define the formulae
\begin{align*}
\varphi_{\mi{Dom}}(u)\ := &\ \exists x y\big(P(x)\wedge P(y)\wedge\beta^*(p_0,x,p_x)\wedge\beta^*(p_0,y,p_y)\wedge\beta^*(x,u,p_y)\wedge\beta^*(y,u,p_x)\big)\\
&\vee \Bigl(u\neq p_x\wedge u\neq p_y\wedge P(u)\wedge \bigl(\beta(p_0,u,p_x)\vee\beta(p_0,u,p_y)\bigr) \Bigr),\\
%
%\varphi_\mi{Frame}(u) := &\ \varphi_{\mi{Dom}}(u)\vee \Bigl( P(u)\wedge \bigl(\beta(p_0,u,p_x)\vee\beta(p_0,u,p_y)\bigr),
%\Bigr),\\
%&\ u\not= p_e\wedge u\not= q_e\\
%
%& \wedge \Big(Pu\vee Qu \vee\ \exists x y\big(Px\wedge x\not= p_e
%
%\wedge Qy\wedge y\not= q_e\wedge \beta(x,u,q_e)\wedge\beta(y,u,p_e)\big)\Big),\\
%
\varphi_{H}(u,v)\ :=&\ 
\exists x \bigl(\beta(p_0,x,p_y)\, \wedge\, \beta(x,u,v)\, \wedge\, \beta^*(u,v,p_x)\bigr)
\wedge\, \forall r\bigl(\, \beta^*(u,r,v)\, \rightarrow\, \neg \varphi_{\mi{Dom}}(r)\, \bigr),\\
\varphi_{V}(u,v)\ :=&\ 
\exists x \bigl(\beta(p_0,x,p_x)\, \wedge\, \beta(x,u,v)\, \wedge\, \beta^*(u,v,p_y)\bigr)
\wedge\, \forall r\bigl(\, \beta^*(u,r,v)\, \rightarrow\, \neg \varphi_{\mi{Dom}}(r)\, \bigr).
\end{align*}
Next we define the following auxiliary formula:
\begin{align*}
\mi{diagonal}(u,v):=&
\exists x\bigl( \varphi_\mi{Dom}(x)\wedge \varphi_\mi{H}(u,x)\wedge\varphi_\mi{V}(x,v)\bigr).
\end{align*}
Recall the function $N$ that associates each tile type $t$ with the unique positive integer $N(t)$.
Let $\exists^{=N(t)}x$ denote the quantifier stating that there exist exactly $N(t)$ $x$:s.
Now, for each tile symbol $P_t$, we define
\begin{align*}
\varphi_{P_t}(u)\ :=&\ 
\exists z \exists^{=N(t)} x \bigl(\, 
\varphi_\mi{Dom}(z)\wedge\mi{diagonal}(u,z) \wedge P(x)\wedge \beta^*(u,x,z)\,  \bigr).
\end{align*}
The formulae
$\varphi_{\mi{Dom}}$, $\varphi_\mi{H}$, $\varphi_\mi{V}$ and $\varphi_{P(t)}$
define the uniform interpretation $I$.
Call $D_\mf{C} := \{\ u\in T\ |\ \mf{C}\models\varphi_{\mi{Dom}}(u)\ \}$, and 
define the
structure
$$\mf{D}_\mf{C} = (D_\mf{C},H^{\mf{D}_\mf{C}},V^{\mf{D}_\mf{C}},(P_t^{\mf{D}_\mf{C}})_{P_t\in S}),$$
where 
\begin{align*}
H^{\mf{D}_\mf{C}}\ :=&\ \{\ (u,v)\in D_\mf{C}\times D_\mf{C}\ |\ \mathfrak{C}\models\varphi_H(u,v)\ \}, \\
V^{\mf{D}_\mf{C}}\ :=&\ \{\ (u,v)\in D_\mf{C}\times D_\mf{C}\ |\ \mathfrak{C}\models\varphi_V(u,v)\ \},\\
P_t^{\mf{D}_\mf{C}}\ :=&\ \{\ u\in D_\mf{C}\ |\ \mathfrak{C}\models\varphi_{P_t}(u)\ \}\text{ for each }P_t\in S.
\end{align*}
By Lemma \ref{nonstandardmodelsexcluded}, it is easy to see that there
is and $S$-labelled grid
$(G,H,V,(P_t)_{P_t\in S})$ such that there
exists an injection $f$ from the domain of the grid to $D_{\mf{C}}$ such that 
the following three conditions hold for all $u,v\in G$ and $P_t\in S$:
\begin{enumerate}
\item
$(u,v)\in H\ \Leftrightarrow\ \varphi_{H}\bigl(f(u),f(v)\bigr)$,
\item
$(u,v)\in V\ 
\Leftrightarrow\ \varphi_{V}\bigl(f(u),f(v)\bigr)$,
\item
$u\in P_{t}\ 
\Leftrightarrow\ \varphi_{P_t}\bigl( f(u)\bigr)$.
\end{enumerate}
Hence there is an $S$-labelled supergrid $\mf{G}_\mf{C}=(G_\mf{C},H^{\mf{G}_\mf{C}},V^{\mf{G}_\mf{C}},(P_t^{\mf{G}_\mf{C}})_{P_t\in S})$ such that
there exists a bijection $f$
from $G_\mf{C}$ to $D_\mf{C}$ such that
the following conditions hold for all $u,v\in G_{\mf{C}}$ and $P_t\in S$:
\begin{enumerate}
\item
$(u,v)\in H^{\mf{G}_\mf{C}}\ \Leftrightarrow\ \varphi_{H}\bigl(f(u),f(v)\bigr)$,
\item
$(u,v)\in V^{\mf{G}_\mf{C}}\ 
\Leftrightarrow\ \varphi_{V}\bigl(f(u),f(v)\bigr)$,
\item
$u\in P^{\mf{G}_\mf{C}}_{t}\ 
\Leftrightarrow\ \varphi_{P_t}\bigl( f(u)\bigr)$.
\end{enumerate}
Let 
$$\mathcal{G}(T,S):=\{\ \mf{G}_{\mf{C}}\in\mathcal{G}_S
\ |\ \mf{C}\in\mathcal{C}(T,S)\ \}.$$
Since $T$ extends linearly in $\mathrm{2D}$, we have $\mf{G}_S\subseteq\mathcal{G}(T,S)$.
The function $I$ is a uniform interpretation of $\mathcal{G}(T,S)$ in $\mathcal{C}(T,S)$.
\end{proof}
\begin{lem}\label{recurrencegridinterpr}
Let $n\geq 2$ be an integer.
There is a computable function $K$ such that for each
nonempty set $S$ of tile symbols, the function $K$ is a 
uniform interpretation of the class of $S$-labelled recurrence grids in the class of
$S$-labelled Cartesian frames with domain $\mathbb{R}^n$.
\end{lem}
\begin{proof}
Straightforward by Lemma \ref{nonstandardmodelsexcluded} and the proof of Lemma \ref{gridinterpretablelemma}.
\end{proof}
\begin{thm}\label{infinitelemma}
Let $T\subseteq \rsq$ be a set and $\beta$ be the corresponding betweenness relation. Assume that $T$ extends linearly in $\mathrm{2D}$.
The first-order theory of the unary expansion class of 
$(T,\beta)$ is $\Sigma_1^0$-hard.
\end{thm}
\begin{proof}
Since $T$ extends linearly in $\mathrm{2D}$, we have $n\geq 2$.
Let $\sigma = \{H,V\}$ be the vocabulary of
supergrids, and let $\tau = \{\beta,P,p_0,p_x,p_y\}$ be the vocabulary of
labelled Cartesian frames.
By Lemma \ref{boxedcartesianpicturesdefinablelemma},
there is a computable function that associates each input $S$ to the tiling problem
with a first-order $\tau$-sentence that defines the class of $S$-labelled Cartesian frames with
the domain $T$ with respect to the class of all
expansions of $(T,\beta)$ to
the vocabulary $\tau$. Let $\varphi_{\mi{Cf}}^S$ denote
such a sentence.
%
%By Lemma \ref{gridinterpretablelemma}, there is a
%
%computable function that associates each input $S$ to the tiling problem
%
%with a uniform interpretation function $I$.
%
%The  function $I_S$ is a uniform interpretation
%
%of some class $\mathcal{G}(T,S)\supseteq \mf{G}_S$  of
%
%$S$-labelled supergrids in the class $\mathcal{C}(T,S)$ of
%
%all $S$-labelled Cartesian frames with the domain $T$.
%
By Lemma \ref{tilingdefinablelemma}, there is a
computable function that associates each input $S$ to the tiling problem with a
first-order $\sigma\cup S$-sentence $\varphi_{S}$
such that a structure $\mf{A}$ of the vocabulary $\sigma$ is $S$-tilable if and only if there is an
expansion $\mf{A}^*$ of the structure
$\mf{A}$ to the vocabulary $\sigma\cup S$ such that 
$\mf{A}^*\models\varphi_{S}$.
Now recall Lemma \ref{gridinterpretablelemma}.
By Lemma \ref{gridinterpretablelemma}, since $T$ extends linearly in $\mathrm{2D}$,
there exists a computable
function $I$ such that for each input $S$ to the tiling problem,
the function $I$ is a uniform interpretation
of some class $\mathcal{G}(T,S)\supseteq \mf{G}_S$  of
$S$-labelled supergrids in the class $\mathcal{C}(T,S)$ of
all $S$-labelled Cartesian frames with the domain $T$.
Let $S$ be an input to the tiling problem. Define the $\tau$-sentence 
\[
\psi_{S}\, :=\, 
 \varphi_{\mi{Cf}}^S\, \wedge\, I(\, \varphi_{S}\, ).
\]
We will prove that for each
input $S$ to the tiling problem, the following conditions are equivalent.
\begin{enumerate}
\item
There \emph{exists} an expansion $\mf{B}$ of $(T,\beta)$ to the vocabulary $\tau$
such that $\mf{B}\models\psi_S$.
\item
The grid $\mf{G}$ is $S$-tilable.
\end{enumerate}
Thereby we establish that there exists a computable reduction from the
\emph{complement problem} of the tiling problem
to the membership problem of the first-order theory of the unary expansion class of $(T,\beta)$.
Since the tiling problem is $\Pi_1^0$-complete,
its complement problem is $\Sigma_1^0$-complete.\footnote{It is of course a
well-known triviality that the complement $\overline{A}$ of a problem $A$ is $\Sigma^0_1$-hard
if $A$ is $\Pi^0_1$-hard. 
Choose an arbitrary problem $B\in \Sigma^0_1$. 
Therefore $\overline{B}\in\Pi^0_1$.
By the hardness of $A$, there is a computable 
reduction $f$ such that $x\in\overline{B}\Leftrightarrow f(x)\in A$, 
whence $x\in B\Leftrightarrow f(x)\in \overline{A}$.}
Let $S$ be an input to the tiling problem.
Assume first that the grid $\mf{G}$ is $S$-tilable.
Therefore there exists an expansion $\mf{G}^*$ of the grid $\mf{G}$ to the
vocabulary $\{H,V\}\, \cup\, S$
such that $\mf{G}^*\models\varphi_{S}$.
Now since $\mf{G}^*\in\mf{G}_S\subseteq\mathcal{G}(T,S)$, by Lemma \ref{uniforminterpretationlemma} there
exists an $S$-labelled Cartesian frame $\mf{C}$ with the domain $T$ such that
$\mf{C}\models I(\varphi_{S})$. 
Since $\mf{C}$ is an $S$-labelled Cartesian frame, we have $\mf{C}\models \varphi_{\mi{Cf}}^S$.
Therefore $\mf{C}\models \varphi_{\mi{Cf}}^S\wedge I(\varphi_{S})$.
Hence the Cartesian frame $\mf{C}$ is an expansion of $(T,\beta)$
such that $\mf{C}\models\psi_S$.
For the converse, assume that there exists an expansion $\mf{B}$ of $(T,\beta)$ to the
vocabulary $\tau$ 
such that we have $\mf{B}\models\psi_S$. Therefore 
$\mf{B}\models \varphi_{\mi{Cf}}^S$ and $\mf{B}\models I(\varphi_{S})$.
Since $\mf{B}\models \varphi_{\mi{Cf}}^S$, the structure 
$\mf{B}$ is an $S$-labelled Cartesian frame with the domain $T$. Therefore, and since $\mf{B}\models I(\varphi_{S})$,
we conclude by Lemma \ref{uniforminterpretationlemma}
that $\mf{A}\models\varphi_{S}$ for some $S$-labelled supergrid $\mf{A}\in\mathcal{G}(T,S)$.
Thus there exists a supergrid that is $S$-tilable.
Hence the grid $\mf{G}$ is $S$-tilable.
\end{proof}
As a partial converse to Theorem \ref{infinitelemma}, we note that $T$ extending
linearly $\mathrm{1D}$ 
is not a sufficient condition for undecidability of even the
monadic $\Pi_1^1$-theory of $(T,\beta)$. 
For instance, the monadic $\Pi_1^1$-theory of $(\mathbb{R},\beta)$ is decidable; this follows
trivially from the known result that the monadic $\Pi_1^1$-theory $(\mathbb{R},\leq)$ is decidable, see \cite{Burgess}. Also the
monadic $\Pi_1^1$-theory of $(\mathbb{Q},\beta)$ is decidable
since the $\mathrm{MSO}$-theory of $(\mathbb{Q},\leq)$ is decidable \cite{Rabin:1969}.
\begin{thm}
Let $n\geq 2$ be an integer. The first-order theory of the unary expansion class of $\rp$ is $\Pi_1^1$-hard.
\end{thm}
\begin{proof}
The proof is essentially the same as the proof of Theorem \ref{infinitelemma}. The main difference is that we use Lemma \ref{recurrencegridinterpr}
and interpret a class of labelled recurrence grids instead of a class of labelled supergrids,
and hence obtain a reduction from the
recurrent tiling problem
instead of the ordinary tiling problem. Thereby we establish $\Pi^1_1$-hardness instead of $\Sigma_1^0$-hardness.
Due to the recurrence condition of the recurrent tiling problem, the result of Lemma \ref{nonstandardmodelsexcluded} that
there is an isomorphism from $(\mathbb{N},\mi{succ})$ to $(P,E)$---rather than an embedding---is essential.
\end{proof}
\begin{cor}
Let $n\geq 2$ be an integer. The monadic $\Pi_1^1$-theory of $(\rsq,\beta)$ is not arithmetical.
\end{cor}

\section{Geometric structures $(T,\beta)$ with an undecidable expansion class with a finite
unary predicate}\label{section5}
In this section we establish undecidability of the
first-order theory of the expansion class
\[
\{(T,\beta, P)\mid P\subseteq T \text{ is finite}\}
\]
of any geometric structure $(T,\beta)$ such that $T$ extends linearly in $\mathrm{2D}$.
More precisely, we show that any such theory is $\Pi^0_1$-hard.
We prove this by a reduction from the periodic
tiling problem to the problem of deciding satisfiability of $\{\beta,P,p_0,p_x.p_y\}$-sentences
in the class of expansions of $(T,\beta)$ by a finite unary predicate $P$ and constants $p_0,p_x,p_y$.
The argument is based on interpreting tori in $(T,\beta)$.
Most notions used in this section are inherited either directly or with minor adjustments from Section \ref{main}.
Let $Q$ be a subset of $T\subseteq\rsq$. We say that $Q$ is a \emph{finite sequence} in $T$ if $Q$ is a finite
nonempty set and the points in $Q$ are collinear.

\begin{defi}
Let $T\subseteq\rsq$ and let $\beta$ be the corresponding betweenness relation.
Let $P\subseteq T$ be a finite set, and let $p_0,p_x,p_y\in P$. We call the structure
\[
\mf{C}=(T,\beta,P,p_0,p_x,p_y)
\]
\emph{a finite Cartesian frame} with domain $T$ if the following conditions hold.
\begin{enumerate}
\item The points $p_0$, $p_x$ and $p_y$ are not collinear.
\item For each point $p\in P$ and $q\in P$ such that $\beta^*(p_0,p,p_x)$
and $\beta^*(p_0,q,p_y)$ hold in $\mf{C}$, the unique lines $L_p$ and $L_q$ in $T$ such that $p,p_y\in L_p$ and $q,p_x\in L_q$, intersect. In other words, there exists a point $u\in T$ that lies on both lines $L_p$ and $L_q$.
\end{enumerate}
If $m$ and $k$ are positive integers such that 
\begin{align*}
&\lvert \{u\in P\mid \beta(p_0,u,p_x) \text{ holds in }\mf{C}\}\rvert=m+2\text{ and}\\
&\lvert \{u\in P\mid \beta(p_0,u,p_y) \text{ holds in }\mf{C}\}\rvert=k+2,
\end{align*}
we call $(T,\beta,P,p_0,p_x,p_y)$ an \emph{$m\times k$ Cartesian frame} with domain $T$.
\end{defi}
\begin{defi}
Let $\mf{C}=(T,\beta,P,p_0,p_x,p_y)$ be a finite Cartesian frame. 
Let $p,q\in P$, $p\neq p_x$, $q\neq p_y$, be points such that $\beta(p_0,p,p_x)$ and $\beta(p_0,q,p_y)$ hold in $\mf{C}$.
Let $L_p$ and $L_q$ be the lines in $T$ such that $p,p_y\in L_p$ and $q,p_x\in L_q$.
The point $u\in T$ that lies on both lines $L_p$ and $L_q$ is called the
\emph{intersection point of $\mf{C}$ corresponding to the pair $(p,q)$}. A point $u\in T$ is called an \emph{intersection point} of the finite Cartesian
frame $\mf{C}$, if it is an intersection point of $\mf{C}$
corresponding to some pair $(p,q)$.
\end{defi}
\begin{defi}
Let $\mf{C}=(T,\beta,P,p_0,p_x,p_y)$ be a finite Cartesian frame.
Let $p,p',q,q'\in P$ be points such that the following conditions hold.
\begin{enumerate}
\item $\beta(p_0,p,p')$ and $\beta^*(p,p',p_x)$ hold in $\mf{C}$.
\item $\beta(p_0,q,q')$ and $\beta^*(q,q',p_y)$ hold in $\mf{C}$.
\item There does not exist a point $u\in P$ such that $\beta^*(p,u,p')$ or $\beta^*(q,u,q')$ holds in $\mf{C}$.
\end{enumerate}
Let $u$ be the intersection point of $\mf{C}$ corresponding to $(p,q)$ and $v$ the
intersection point of $\mf{C}$ correponding to $(p',q')$.
We say that $v$ is the \emph{diagonal successor} of $u$ in $\mf{C}$.
\end{defi}
\begin{defi}
Let $\mf{C}=(T,\beta,P,p_0,p_x,p_y)$ be a finite Cartesian frame and let $S$ be a finite nonempty set of tile symbols.
We call $\mf{C}$ an \emph{$S$-labelled finite Cartesian frame} if the number of points in $P$ strictly between
any intersection point $u$ of $\mf{C}$ and its diagonal successor $v$ is in the set $\{N(P_t)\mid P_t\in S\}$.
We let
$\mathcal{C}^{fin}(T,S)$ denote \emph{the class of $S$-labelled finite Cartesian frames with domain $T$}.
\end{defi}

\begin{lem}\label{finitedefinable}
Let $T\subseteq \rsq$, $n\geq 2$.  Let $\mathcal{C}$ be the class of all expansions $(T,\beta,P,p_0,p_x,p_y)$ of $(T,\beta)$ by a finite
unary relation $P$ and constants $p_0$, $p_x$ and $p_y$. There is a computable function associating each finite nonempty set of tile symbols $S$ with a first-order $\{\beta,P,p_0,p_x,p_y\}$-sentence $\varphi_{fCf}^S$ 
that defines the class $\mathcal{C}^\mi{fin}(T,S)$ with respect to the class $\mathcal{C}$.
\end{lem}
\begin{proof}
Straightforward.
\end{proof}
Let $S\not=\emptyset$ be a finite set of tile symbols.
Let $\mathcal{T}_S$ denote the class of structures $\mf{A}$ that satisfy the following conditions.
\begin{enumerate}
\item The structure $\mf{A}$ is an expansion of some torus to the vocabulary $\{V,H\}\cup S$.
\item Each point in the domain of $\mf{A}$ is in the extension of exactly one predicate symbol $P_t\in S$.
\end{enumerate} 
Structures in the class $\mathcal{T}_S$ are called \emph{$S$-labelled tori}.
\begin{lem}\label{torusinterpretablelemma}
Let $T\subseteq\mathbb{R}^n$, $n\geq 2$. Assume that $T$ extends linearly in $\mathrm{2D}$.
There is a computable function $J$ such that for all finite sets $S\not=\emptyset$ of tile symbols, $J$ is a
uniform interpretation of the class of $S$-labelled tori in
 $\mathcal{C}^\mi{Fin}(T,S)$.
\end{lem}
\begin{proof}
Let $S$ be a finite nonempty set of tile symbols. Let $\mf{C}=(T,\beta,P,p_0,p_x,p_y)$
be an $S$-labelled $m\times k$ Cartesian frame. We shall show how to interpret an $S$-labelled $m\times k$ torus
 $\mf{T}_\mf{C}$ in $\mf{C}$. The idea behind the
interpretation is quite similar to the
idea behind the interpretation in the proof of Lemma \ref{gridinterpretablelemma}.
Recall the formulae $\varphi_{\mi{Dom}}$, $\varphi_{H}$, $\varphi_{V}$ and $\varphi_{P_t}$
defined in the proof of Lemma \ref{gridinterpretablelemma}. We shall now define variants of
these formulae suitable for interpreting $S$-labelled tori in $S$-labelled finite Cartesian frames.
In the definitions of the new formulae, we shall make use of the old formulae
$\varphi_{\mi{Dom}}$, $\varphi_{H}$, $\varphi_{V}$ and $\varphi_{P_t}$.
We define
\begin{align*}
\varphi^{\mi{fin}}_{\mi{Dom}}(u)&\ :=\ \varphi_{\mi{Dom}}(u)\wedge \exists x\exists
y\bigl(\varphi_{\mi{Dom}}(x)\wedge \varphi_{\mi{Dom}}(y)\wedge \varphi_H(u,x)\wedge \varphi_V(u,y)  \bigr),\\
\varphi^{\mi{fin}}_H(u,v)&\ :=\ \varphi_H(u,v)
\vee\Big(\beta(p_0,v,p_y)\wedge\beta(v,u,p_x)\wedge
\forall x\big(\beta^*(u,x,p_x)\rightarrow\neg\varphi_{\mi{Dom}}^{\mi{fin}}(x)\big)\Big),\\
\varphi^{\mi{fin}}_V(u,v)&\ :=\ \varphi_V(u,v)
\vee\Big(\beta(p_0,v,p_x)\wedge\beta(v,u,p_y)\wedge
\forall x\big(\beta^*(u,x,p_y)\rightarrow\neg\varphi_{\mi{Dom}}^{\mi{fin}}(x)\big)\Big),\\
\varphi^\mi{fin}_{P_t}(u)&\ :=\ \varphi_{P_t}(u)\text{ for each }P_t\in S.
\end{align*}
Let $F_\mf{C}:=\{u\in T \mid \mf{C}\models \varphi^{\mi{fin}}_{Dom}(u)\}$. Define the structure
$$\mf{F}_\mf{C}=\bigl(F_\mf{C},H^{\mf{F}_\mf{C}},V^{\mf{F}_\mf{C}}, (P_t^{\mf{F}_\mf{C}})_{P_t\in S}\bigr),$$
where
\begin{align*}
H^{\mf{F}_\mf{C}}&:=\{(u,v)\in F_\mf{C}\times F_\mf{C}\mid \mf{C}\models \varphi^{\mi{fin}}_H(u,v)\}, \\
V^{\mf{F}_\mf{C}}&:=\{(u,v)\in F_\mf{C}\times F_\mf{C}\mid \mf{C}\models \varphi^{\mi{fin}}_V(u,v)\}, \\
P_t^{\mf{F}_\mf{C}}&:=\{u\in F_\mf{C}\mid \mf{C}\models \varphi^\mi{fin}_{P_t}(u)\},
\end{align*}
for all $P_t\in S$.
It is straightforward to check that there exists an $S$-labelled $m\times k$ torus
\[ \mf{T}_{\mf{C}}=\bigl(D,H^{\mf{T}_\mf{C}},V^{\mf{T}_\mf{C}}, (P_t^{\mf{T}_\mf{C}})_{P_t\in S}\bigr)
\]
and a bijection $f$ from $D$ to $F_\mf{C}$ such that the following conditions hold for all $u,v\in D$.
\begin{enumerate}
\item $(u,v)\in H^{\mf{T}_\mf{C}}\ \Leftrightarrow\ \varphi^{\mi{fin}}_H(f(u),f(v))$,
\vspace{1pt}
\item $(u,v)\in V^{\mf{T}_\mf{C}}\ \Leftrightarrow\ \varphi^{\mi{fin}}_V(f(u),f(v))$,
\vspace{1pt}
\item $u\in P_t^{\mf{T}_\mf{C}}\ \Leftrightarrow\ \varphi^\mi{fin}_{P_t}(f(u))$ for all $P_t\in S$.
\end{enumerate}
Notice that since $T$ extends linearly in $\mathrm{2D}$, there exist $S$-labelled finite Cartesian frames of all sizes $m\times k$
with all possible $S$-labelling configurations in the class
$\mathcal{C}^{\mi{fin}}(T,S)$.
We have hence established that for all finite sets $S\not=\emptyset$ of tile symbols,
the class of $S$-labelled tori is uniformly first-order interpretable in the class of $S$-labelled finite Cartesian frames
 with the domain $T$. Furthermore, the
formulae $\varphi_{\mi{Dom}}^\mi{fin}$, $\varphi_H^\mi{fin}$, $\varphi_V^\mi{fin}$ and $\varphi_{P_t}^\mi{fin}$
define the desired uniform interpretation $J$.
\end{proof}
\begin{thm}
Let $T\subseteq \rsq$ and let $\beta$ be the corresponding betweenness relation.
Assume that $T$ extends linearly in $\mathrm{2D}$. The
first-order theory of the class $\{(T,\beta,P)\mid P\subseteq T\text{ is finite}\}$ is $\Pi^0_1$-hard.
\end{thm}
\begin{proof}
Since $T$ extends linearly in $\mathrm{2D}$, we have $n\geq 2$.
Let $\sigma=\{H,V\}$ be the vocabulary of tori, and let $\tau=\{\beta,P,p_0,p_x,p_y\}$ be the vocabulary of labelled
finite Cartesian frames.
By Lemma \ref{finitedefinable},
there is a computable function that associates each input $S$ to the periodic tiling problem with a first-order $\tau$-sentence that
defines the class of $S$-labelled finite Cartesian frames with the domain $T$ with respect to the
class of all expansions of $(T,\beta)$ to the vocabulary $\tau$. Let $\varphi^S_{\mi{fCf}}$ denote such a sentence.
By Lemma \ref{torusinterpretablelemma}, there is a computable function $J$ such that 
for all inputs $S$ to the periodic tiling problem, the function $J$ 
is a uniform interpretation of the class of $S$-labelled tori in the class of $S$-labelled 
finite Cartesian frames with domain $T$.
By Lemma \ref{tilingdefinablelemma}, there  is a computable function that associates each input $S$ to the periodic tiling problem
with a first-order $\sigma\cup S$-sentence $\varphi_S$ such that
for all tori $\mf{B}$, the torus $\mf{B}$ is $S$-tilable iff there is an
expansion $\mf{B}^*$ of $\mf{B}$ to the
vocabulary $\sigma\cup S$
such that $\mf{B}^*\models\varphi_{S}$.

Let $S$ be a finite nonempty set of tile symbols. Define the first-order $\tau$-sentence
\[
\gamma_S:=\varphi_{\mi{fCf}}^S\wedge J(\varphi_S).
\]
We will prove that for each input $S$ to the periodic tiling problem, the following conditions are equivalent.
\begin{enumerate}
\item There exists an expansion $\mf{B}=(T,\beta,P,p_0,p_x,p_y)$ of $(T,\beta)$ by a finite unary relation $P\subseteq T$ and constants $p_0,p_x,p_y\in T$ such that  $\mf{B}\models\gamma_S$.
\item There exists a torus $\mf{T}$ such that $\mf{T}$ is $S$-tilable.
\end{enumerate}
Thereby we establish that there exists a computable reduction from the complement problem of the periodic tiling problem to the membership problem of the first-order theory of the unary expansion class of $(T,\beta)$ with a
finite predicate. Since the periodic tiling problem is $\Sigma^0_1$-complete, its complement problem is $\Pi^0_1$-complete.

Let $S$ be an input to the periodic tiling problem. First assume that there exists a torus $\mf{T}$ such that $\mf{T}$ is $S$-tilable.
Therefore, by Lemma \ref{tilingdefinablelemma}, there exists an expansion $\mf{T}^*$ of $\mf{T}$ to the vocabulary $\sigma\cup S$ such that
$\mf{T}^*\models \varphi_S$. Since the function $J$ is a uniform interpretation of
the class of $S$-labelled tori in the class of $S$-labelled
finite Cartesian frames with the domain $T$, and since
 $\mf{T}^*\models\varphi_S$,
it follows by Lemma \ref{uniforminterpretationlemma} that there exists an $S$-labelled finite Cartesian frame $\mf{C}$ with the
domain $T$ such that $\mf{C}\models J(\varphi_S)$.
Since $\mf{C}$ is an $S$-labelled finite Cartesian frame,
we have that $\mf{C}\models\varphi^S_\mi{fcf}$.
Therefore $\mf{C}\models\varphi^S_\mi{fcf}\wedge J(\varphi_S)$. Hence the
finite Cartesian frame $\mf{C}$ is an
expansion of $(T,\beta)$ by a finite unary relation $P\subseteq T$ and constants $p_0,p_x,p_y\in T$ such that
 $\mf{C}\models\gamma_S$.

For the converse, assume that there exists an expansion $\mf{B}=(T,\beta,P,p_0,p_x,p_y)$ of $(T,\beta)$ by a finite unary relation $P\subseteq T$ and constants $p_0,p_x,p_y\in T$ such that $\mf{B}\models\gamma_S$.
Therefore $\mf{B}\models\varphi^S_\mi{fCf}$ and $\mf{B}\models J_S(\varphi_S)$. Since $\mf{B}\models\varphi^S_\mi{fCf}$, the structure $\mf{B}$ is an
$S$-labelled finite Cartesian frame with domain $T$. Therefore, and since $\mf{B}\models J(\varphi_S)$, we conclude by Lemma \ref{uniforminterpretationlemma} that $\mf{A}\models\varphi_S$ holds
for some $S$-labelled torus
$\mf{A}$.
Hence by Lemma \ref{tilingdefinablelemma}
there exists a torus that is $S$-tilable.
\end{proof}
\section{Conclusions}
We have studied first-order theories of unary expansion classes of 
geometric structures $(T,\beta)$, $T\subseteq\rsq$.
We have established that for $n\geq 2$, the first-order theory of the class of all expansions of $\rp$ with a single unary
predicate is
highly undecidable ($\Pi_1^1$-hard). This refutes a conjecture from the article \cite{vanBenthemAiello:2002} of Aiello and van Benthem.
In addition, we have established the following for any geometric structure $(T,\beta)$ that extends linearly in $\mathrm{2D}$.
\begin{enumerate}
\item
The first-order theory of the class of expansions of $(T,\beta)$ with a single unary predicate is $\Sigma_1^0$-hard.
\item
The first-order theory of the class of expansions of $(T,\beta)$ with a single finite unary predicate is $\Pi_1^0$-hard.
\end{enumerate}
Geometric structures that extend linearly in $\mathrm{2D}$ include, for example, the rational plane $(\mathbb{Q}^2,\beta)$ and
the real unit
rectangle $([0,1]^2,\beta)$, to name a few.
The techniques used in the proofs can be easily modified to yield undecidability of first-order theories of a significant variety of
natural
restricted expansion classes of the affine real plane $(\mathbb{R}^2,\beta)$, such as those
with a unary predicate denoting a
 polygon, a finite union of closed rectangles, and a semialgebraic set, for example.
Such classes could be interesting from the point of view of applications.

In addition to studying issues of decidability, we briefly compared the expressivities of universal monadic second-order logic and
weak universal monadic second-order logic. While the two are incomparable in general, we established that over any class of
expansions of $\rp$, it is no longer the case. We showed that finiteness of a unary predicate is definable by a first-order sentence,
and hence obtained translations from $\forall\mathrm{WMSO}$ into $\forall\mathrm{MSO}$ and from $\mathrm{WMSO}$ into
$\mathrm{MSO}$.

Our original objective to study expansion classes of $\rp$ was to identify decidable logics of space with distinguished regions.
%This objective of course failed due to the undecidability result we obtained.
Due to the ubiquitous applicability of the tiling methods, this pursuit gave way to identifying several undecidable theories of geometry.
Hence we shall look elsewhere in order to identify well behaved natural decidable logics of space. Possible interesting directions include considering natural fragments of first-order logic over expansions of $\rp$, and also other geometries.
Related results could provide insight, for example, in the background theory of modal spatial logics.

%%
%% Bibliography
%%

%% Either use bibtex (recommended), but commented out in this sample

\bibliographystyle{plain}
\bibliography{tampere}

\begin{thebibliography}{10}

\bibitem{AielloPrattHartmannvanBenthem:2007}
Marco Aiello, Ian Pratt-Hartmann, and Johan Benthem.
\newblock What is spatial logic?
\newblock In Marco Aiello, Ian Pratt-Hartmann, and Johan van Benthem, editors,
  {\em Handbook of Spatial Logics}, pages 1--11. Springer Netherlands, 2007.

\bibitem{handbookofspatial}
Marco Aiello, Ian Pratt-Hartmann, and Johan van Benthem, editors.
\newblock {\em Handbook of Spatial Logics}.
\newblock Springer, 2007.

\bibitem{vanBenthemAiello:2002}
Marco Aiello and Johan van Benthem.
\newblock A modal walk through space.
\newblock {\em Journal of Applied Non-classical Logics}, 12:319--363, 2002.

\bibitem{Balbianietal:1997}
Philippe Balbiani, Luis~Fariñas del Cerro, Tinko Tinchev, and Dimiter
  Vakarelov.
\newblock Modal logics for incidence geometries.
\newblock {\em Journal of Logic and Computation}, 7(1):59--78, 1997.

\bibitem{BalbianiGoranko:2002}
Philippe Balbiani and Valentin Goranko.
\newblock Modal logics for parallelism, orthogonality, and affine geometries.
\newblock {\em Journal of Applied Non-Classical Logics}, 12(3-4):365--398,
  2002.

\bibitem{Balbianietal:2007}
Philippe Balbiani, Valentin Goranko, Ruaan Kellerman, and Dimiter Vakarelov.
\newblock Logical theories for fragments of elementary geometry.
\newblock In Marco Aiello, Ian Pratt-Hartmann, and Johan van Benthem, editors,
  {\em Handbook of Spatial Logics}, pages 343--428. Springer, 2007.

\bibitem{Berger}
Robert Berger.
\newblock {\em The Undecidability of the Domino Problem}.
\newblock American Mathematical Society memoirs. American Mathematical Society,
  1966.

\bibitem{realalgebraicgeometry}
Jacek Bochnak, Michel Coste, and Marie-Francoise Roy.
\newblock {\em Real Algebraic Geometry}.
\newblock Springer, 1998.

\bibitem{Burgess}
John~P. Burgess and Yuri Gurevich.
\newblock The decision problem for linear temporal logic.
\newblock {\em Notre Dame Journal of Formal Logic}, 26(2):115--128, April 1985.

\bibitem{tenCate:2011}
Balder~ten Cate and Alessandro Facchini.
\newblock Characterizing {EF} over infinite trees and modal logic on transitive
  graphs.
\newblock In Filip Murlak and Piotr Sankowski, editors, {\em MFCS}, volume 6907
  of {\em Lecture Notes in Computer Science}, pages 290--302. Springer, 2011.

\bibitem{Griffiths:2008}
Aled Griffiths.
\newblock {\em Computational Properties of Spatial Logics in the Real Plane}.
\newblock PhD thesis, University of Manchester, 2008.

\bibitem{gure1}
Yuri Gurevich.
\newblock Monadic second-order theories.
\newblock In Jon Barwise and Solomon Feferman, editors, {\em Model-Theoretic
  Logics}, pages 479--506. Springer, New York, 1985.

\bibitem{GurevichKoryakov:1972}
Yuri Gurevich and Igor~O. Koryakov.
\newblock Remarks on berger's paper on the domino problem.
\newblock {\em Siberian Mathematical Journal}, 13:319--321, 1972.

\bibitem{Gyssens:1999}
Marc Gyssens, Jan~Van den Bussche, and Dirk~Van Gucht.
\newblock Complete geometric query languages.
\newblock {\em Journal of Computer and System Sciences}, 58(3):483--511, 1999.

\bibitem{Harel:1985}
David Harel.
\newblock Recurring dominoes: Making the highly undecidable highly
  understandable.
\newblock {\em Annals of Discrete Mathematics}, 24:51--72, 1985.

\bibitem{HodkinsonHussain}
Ian Hodkinson and Altaf Hussain.
\newblock The modal logic of affine planes is not finitely axiomatisable.
\newblock {\em Journal of Symbolic Logic}, 73(3):940--952, 2008.

\bibitem{KoPrWoZa:2010}
Roman Kontchakov, Ian Pratt-Hartmann, Frank Wolter, and Michael Zakharyaschev.
\newblock Spatial logics with connectedness predicates.
\newblock {\em Logical Methods in Computer Science}, 6(3), 2010.

\bibitem{KujpersVandenBussche}
Bart Kuijpers and Jan Van~den Bussche.
\newblock {\em Logical aspects of spatial databases.}, pages 77--108.
\newblock Cambridge: Cambridge University Press, 2011.

\bibitem{KMV:2012}
Antti Kuusisto, Jeremy Meyers, and Jonni Virtema.
\newblock {Undecidable First-Order Theories of Affine Geometries}.
\newblock In Patrick C{\'e}gielski and Arnaud Durand, editors, {\em Computer
  Science Logic (CSL'12) - 26th International Workshop/21st Annual Conference
  of the EACSL}, volume~16 of {\em Leibniz International Proceedings in
  Informatics (LIPIcs)}, pages 470--484, Dagstuhl, Germany, 2012. Schloss
  Dagstuhl--Leibniz-Zentrum fuer Informatik.

\bibitem{Libkin}
Leonid Libkin.
\newblock {\em Elements of Finite Model Theory}.
\newblock Springer, 2004.

\bibitem{Nenov:2010}
Yavor Nenov and Ian Pratt-Hartmann.
\newblock On the computability of region-based euclidean logics.
\newblock In Anuj Dawar and Helmut Veith, editors, {\em CSL}, volume 6247 of
  {\em Lecture Notes in Computer Science}, pages 439--453. Springer, 2010.

\bibitem{pambuccian}
Victor Pambuccian.
\newblock The axiomatics of ordered geometry: I. ordered incidence spaces.
\newblock {\em Expositiones Mathematicae}, 29(1):24 -- 66, 2011.

\bibitem{prestel1}
Alexander {Prestel}.
\newblock {Zur Axiomatisierung gewisser affiner Geometrien.}
\newblock {\em {L'Enseignement Math\'ematique (2)}}, 27:125--136, 1981.

\bibitem{prestel2}
Alexander {Prestel} and Leslaw~W. {Szczerba}.
\newblock {Nonaxiomatizability of real general affine geometry.}
\newblock {\em {Fundamenta Mathematicae}}, 104:193--202, 1979.

\bibitem{Rabin:1969}
Michael~O. Rabin.
\newblock {\em Decidability of Second-order Theories and Automata on Infinite
  Trees}.
\newblock IBM Watson Research Center, 1968.

\bibitem{Shere:2010}
Mikhail Sheremet, Frank Wolter, and Michael Zakharyaschev.
\newblock A modal logic framework for reasoning about comparative distances and
  topology.
\newblock {\em Annals of Pure and Applied Logic}, 161(4):534--559, 2010.

\bibitem{schw}
Wolfram {Schwabh\"auser};~Wanda {Szmielew} and Alfred {Tarski}.
\newblock {\em {Metamathematical methods in geometry. Part I: An axiomatic
  building of Euclidean geometries. Part II: Metamathematical considerations.
  With a new foreword by Michael Beeson. (Metamathematische Methoden in der
  Geometrie. Teil I: Ein axiomatischer Aufbau der euklidischen Geometrie. Teil
  II: Metamathematische Betrachtungen.) Reprint of the 1983 original published
  by Springer.}}
\newblock Bronx, NY: Ishi Press International, reprint of the 1983 original
  published by springer edition, 2011.

\bibitem{Tarski:1948}
Alfred Tarski.
\newblock A decision method for elementary algebra and geometry.
\newblock {\em University of California Press, Berkeley}, 5, 1951.

\bibitem{TarskiGivant:1999}
Alfred Tarski and Steven Givant.
\newblock Tarski's system of geometry.
\newblock {\em The Bulletin of Symbolic Logic}, 5(2):175--214, 1999.

\bibitem{Tinchev}
Tinko Tinchev and Dimiter Vakarelov.
\newblock Logics of space with connectedness predicates: Complete
  axiomatizations.
\newblock In Lev~D. Beklemishev, Valentin Goranko, and Valentin Shehtman,
  editors, {\em Advances in Modal Logic}, pages 434--453. College Publications,
  2010.

\bibitem{Venema:1999}
Yde Venema.
\newblock Points, lines and diamonds: A two-sorted modal logic for projective
  planes.
\newblock {\em Journal of Logic and Computation}, 9(5):601--621, 1999.

\end{thebibliography}


\begin{thebibliography}{30}
\bibitem{vanBenthemAiello:2002}
M. Aiello and J. van Benthem. A Modal Walk through Space.
\emph{Journal of Applied Non-Classical Logics} 12(3-4):319-363, Hermes, 2002.

\bibitem{AielloPrattHartmannvanBenthem:2007}
M. Aiello, I. Pratt-Hartmann, and J. van Benthem.
What is Spatial Logic. In Marco Aiello, Ian Pratt-Hartmann and Johan van Benthem,
\emph{Handbook of Spatial Logics}, Springer, 2007.

\bibitem{handbookofspatial}
M. Aiello, I. Pratt-Hartmann and J. van Benthem. Handbook of Spatial Logics.
Springer, 2007.

\bibitem{Balbianietal:1997}
P. Balbiani, L. Farinas del Cerro, T. Tinchev, and D. Vakarelov. Modal Logics for Incidence Geometries, \emph{Journal of Logic and Computation}, 7(1), 59-78, 1997.

\bibitem{Balbiani and Goranko:2002}
P. Balbiani and V. Goranko.
Modal logics for parallelism, orthogonality, and affine geometries,
\emph{Journal of Applied Non-Classical Logics}, 12,365-397, 2002.

\bibitem{Balbianietal:2007} 
P. Balbiani, V. Goranko, R. Kellerman and D. Vakarelov. Logical Theories for Fragments of Elementary Geometry. In \emph{Handbook of Spatial Logics}. Springer. 343-428, 2007.

\bibitem{Berger}
R. Berger. The undecidability of the domino problem. \emph{Mem. Amer. Math. Soc.}, 66, 1966.

\bibitem{realalgebraicgeometry}
J. Bochnak, M. Coste and M. Roy.
\emph{Real Algebraic Geometry}, Springer, 1998.

\bibitem{Burgess}
J. P. Burgess and Y. Gurevich. The Decision Problem for Linear Temporal Logic. \emph{Notre Dame Journal of Formal Logic}, vol. 26, no. 2, 1985. 

\bibitem{tenCate:2011}
B. ten Cate and A. Facchini. Characterizing {EF} over Infnite Trees and Modal Logic on Transitive Graphs. \emph{Proceedings of the MFCS}, 2011.

\bibitem{Griffiths:2008}
A. Griffiths. Computational Properties of Spatial Logics in the Real Plane. PhD thesis, University of Manchester, 2008.

\bibitem{GurevichKoryakov:1972}
Y. Gurevich and I. O. Koryakov. Remarks on Berger's paper on the domino problem.
\emph{Siberian Mathematical Journal 13}, 319-321, 1972.

\bibitem{Gyssens:1999}
M. Gyssens, J. Van den Bussche and D. Van Gucht. Complete Geometric Query Languages. \emph{Journal of Computer and System Sciences 58}, 483-511, 1999.

\bibitem{Harel:1985}
D. Harel. Recurring Dominoes: Making the Highly Undecidable Highly Understandable.
\emph{Annals of Discrete Mathematics 24}, 51-72, 1985.


\bibitem{HodkinsonHussain}
I. Hodkinson and A. Hussain. The modal logic of affine planes is not finitely axiomatisable. \emph{Journal of Symbolic Logic} 73(3), 940-952, 2008.

\bibitem{KoPrWoZa:2010}
R. Kontchakov, I. Pratt-Hartmann, F. Wolter and M. Zakharyaschev.
Spatial logics with connectedness predicates.
\emph{Logical Methods in Computer Science}, 6(3), 2010.

\bibitem{KujpersVandenBussche}
B. Kujpers and J. Van den Bussche. Logical aspects of spatial database theory.
In \emph{Finite and Algorithmic Model Theory},
London Mathematical Society Lecture Notes Series 379, Cambridge University Press, 2011.

\bibitem{KMV:2012}
A. Kuusisto, J. Meyers and J. Virtema. Undecidable First-Order Theories of Affine Geometries. In Proceedings of the \emph{21th EACSL Annual Conference on Computer Science Logic}, 2012.

\bibitem{Libkin}
L. Libkin. Elements of Finite Model Theory. Springer-Verlag, 2004.

\bibitem{Nenov:2010}
Y. Nenov and I. Pratt-Hartmann. On the Computability of Region-Based Euclidean Logics. In Proceedings of the \emph{19th EACSL Annual Conference on Computer Science Logic}, 2010.


\bibitem{Rabin:1969}
M. O. Rabin. Decidability of second-order theories and automata on infinite trees.
\emph{Trans. of the Amer. Math. Soc.} 141, 1-35, 1969.

\bibitem{Shere:2010}
M. Sheremet, F. Wolter and M. Zakharyaschev.
A modal logic framework for reasoning about comparative distances and topology.
\emph{Ann. Pure Appl. Logic}, 161(4):534-559, 2010.

\bibitem{Tarski:1948}
A. Tarski. A decision method for elementary algebra and geometry. RAND Corporation, Santa Monica, 1948.

\bibitem{TarskiGivant:1999}
A. Tarski and S. Givant. Tarski's System of Geometry.
\emph{Bull. Symbolic Logic} 5(2), 1999.

\bibitem{Tinchev}
T. Tinchev and D. Vakarelov.
Logics of Space with Connectedness Predicates: Complete Axiomatizations. In
Proceedings of the \emph{Advances in Modal Logic 8 (AiML)}, 434-453, 2010.


\bibitem{Venema:1999}
Y. Venema. Points, Lines and Diamonds: a Two-sorted Modal Logic for Projective Planes. \emph{Journal of Logic and Computation}, 9(5) 601-621, 1999.
\end{thebibliography}

%% .. or use bibitems explicitely

\begin{comment}

\end{comment}
\end{document}